\documentclass[runningheads,a4paper]{llncs}

\usepackage[usenames,dvipsnames,svgnames,table]{xcolor}
\usepackage{microtype}
\usepackage{amsmath,amssymb}
\usepackage{paralist}
\usepackage{comment}
\usepackage[bookmarks,bookmarksopen,bookmarksdepth=2]{hyperref}

\let\llncssubparagraph\subparagraph
\let\subparagraph\paragraph
\usepackage[compact]{titlesec}
\let\subparagraph\llncssubparagraph

\newcommand{\qedhere}{\mbox{\qed}}

\newif\ifAnonymous
\Anonymousfalse

\newif\ifComments
\Commentsfalse

\ifComments
\newcommand{\td}[1]{
\begin{center}
\fbox{
\begin{minipage}{3in} \color{Purple}
{\bf TD-Y:} {\rm #1}
\end{minipage}
}
\end{center}
}
\else
\newcommand{\td}[1]{}
\fi

\ifComments
\newcommand{\nt}[1]{
\begin{center}
\fbox{
\begin{minipage}{3in} \color{Purple}
{\bf NT:} {\rm #1}
\end{minipage}
}
\end{center}
}
\else
\newcommand{\nt}[1]{}
\fi

\ifComments
\newcommand{\lb}[1]{
\begin{center}
\fbox{
\begin{minipage}{3in} \color{Purple}
{\bf LB} {\rm #1}
\end{minipage}
}
\end{center}
}
\else
\newcommand{\lb}[1]{}
\fi

\ifComments
\newcommand{\ks}[1]{
\begin{center}
\fbox{
\begin{minipage}{3in} \color{Purple}
{\bf KS} {\rm #1}
\end{minipage}
}
\end{center}
}
\else
\newcommand{\ks}[1]{}
\fi

\ifComments
\newcommand{\gj}[1]{
\begin{center}
\fbox{
\begin{minipage}{3in} \color{Purple}
{\bf GJ} {\rm #1}
\end{minipage}
}
\end{center}
}
\else
\newcommand{\gj}[1]{}
\fi

\usepackage{listings}
\usepackage{mathpartir}

\newcommand{\ty}[1]{\ensuremath{\mathsf{#1}}}

\newcommand{\valid}{\mathsf{valid}}
\newcommand{\spec}{\mathsf{spec}}
\newcommand{\emp}{\mathsf{emp}}

\newcommand{\wand}{\mathbin{-\hspace*{-3pt}*}}%

\newcommand{\triple}[3]{\left\{ \begin{array}{@{}c@{}} #1 \end{array} \middle\} \ {#2} \ \middle\{ \begin{array}{@{}c@{}}#3 \end{array} \right\}}

\newcommand{\vd}{\vdash}

\newcommand{\vdel}{\vdash}

\newcommand\unit{()}
\newcommand\sep{\mid}
\newcommand\tty[1]{{\tt #1}}
\newcommand\pair[2]{\langle#1,#2\rangle}

\newcommand\emit{\mathsf{emit}}
\newcommand\TEMPletin[2]{\mathsf{let}\ #1\ \mathsf{in}\ #2}
\newcommand\TEMPemit{\mathsf{emit}}
\newcommand\ifte[3]{{\rm if}\, #1\, #2\, #3}
\newcommand\proj[1]{\pi_{#1}}
\newcommand\refer{\mathsf{ref}\,}
\newcommand\hole{\bullet}

\newcommand\Loc{{\rm Loc}}
\newcommand\config[1]{#1}
\newcommand\dom[1]{\mathsf{dom}(#1)}
\newcommand\erase[1]{\widehat{#1}}

\newcommand\wrap{{\bf wrap}}

\newcommand{\nil}{\textsf{nil}}
\newcommand{\cons}{\textsf{cons}}
\newcommand{\fold}{\textsf{fold}}

\newcommand{\true}{\textsf{true}}
\newcommand{\false}{\textsf{false}}

\newcommand{\Tif}{\textsf{if}}
\newcommand{\Tthen}{\textsf{then}}
\newcommand{\Telse}{\textsf{else}}

\newcommand{\eqdef}{\stackrel{\text{\tiny{def}}}{=}}
\newcommand{\Setb}[2]{\left\{#1 \ \middle| \ #2 \right\}}

\newcommand{\leqprefix}{\leq_{\textit{pref}}}

\newcommand{\monarrow}{\rightarrow_{\text{mon}}}

\newcommand{\sem}[1]{[\![#1]\!]}
\newcommand{\semsl}[1]{\sem{#1}}
\newcommand{\semel}[1]{\sem{#1}}

\newcommand{\modelsel}{\models}

\newcommand{\semdomain}[1]{\textit{#1}}

\newcommand{\VAL}{\semdomain{Val}}
\newcommand{\VAR}{\semdomain{Var}}
\newcommand{\EXP}{\semdomain{Exp}}
\newcommand{\LAM}{\semdomain{Lam}}
\newcommand{\LOC}{\semdomain{Loc}}
\newcommand{\HEAP}{\semdomain{Heap}}
\newcommand{\FUN}{\semdomain{Fun}}
\newcommand{\FENV}{\semdomain{FEnv}}
\newcommand{\CON}{\semdomain{Con}}

\newcommand{\collp}{\textsf{coll}}
\newcommand{\iterp}{\textsf{iter}}
\newcommand{\collpw}{\collp_{\textsf{w}}}
\newcommand{\iterpw}{\iterp_{\textsf{w}}}

\newcommand{\wrapcoll}{\wrap_{\textsf{coll}}}
\newcommand{\Phicoll}{\Phi_{\textsf{coll}}}
\newcommand{\Lcoll}{\mathcal{L}_{\textsf{coll}}}

\newcommand{\Rp}{\textsf{R}}
\newcommand{\Rpw}{\textsf{R}_{\textsf{w}}}
\newcommand{\strp}{\textsf{str}}
\newcommand{\strpw}{\textsf{str}_{\textsf{w}}}
\newcommand{\esafep}{\textsf{esafe}}
\newcommand{\safep}{\textsf{safe}}
\newcommand{\safepw}{\textsf{safe}_{\textsf{w}}}

\newcommand{\wrapstr}{\wrap_{\textsf{str}}}
\newcommand{\Phistr}{\Phi_{\textsf{str}}}
\newcommand{\Lstr}{\mathcal{L}_\mathsf{str}}
\newcommand{\validsanip}{\mathsf{str}_{\textsf{trace}}}
\newcommand{\isallocp}{\mathsf{allocs}}
\newcommand{\notfreshp}{\mathsf{notfresh}}

\newcommand{\openp}{\textsf{open}}
\newcommand{\openpw}{\textsf{open}_{\textsf{w}}}
\newcommand{\closedp}{\textsf{closed}}
\newcommand{\closedpw}{\textsf{closed}_{\textsf{w}}}

\newcommand{\unlockedp}{\textsf{unlocked}}
\newcommand{\unlockedpw}{\textsf{unlocked}_{\textsf{w}}}
\newcommand{\lockedp}{\textsf{locked}}
\newcommand{\lockedpw}{\textsf{locked}_{\textsf{w}}}

\newcommand{\wrapbrack}{\wrap_{\textsf{brac}}}
\newcommand{\Phibrack}{\Phi_{\textsf{brac}}}
\newcommand{\Lbrack}{\mathcal{L}_\mathsf{brac}}

\newcommand{\wrapfile}{\wrap_{\textsf{file}}}
\newcommand{\Phifile}{\Phi_{\textsf{file}}}
\newcommand{\Lfile}{\mathcal{L}_\mathsf{file}}

\newcommand{\validfilep}{\mathsf{file}_{\textsf{trace}}}
\newcommand{\noclosep}{\mathsf{noclose}}
\newcommand{\isopenp}{\mathsf{isopen}}

\newcommand{\csharplstset}{%
  \lstset{
    basicstyle = \sffamily,
    columns = fullflexible,
    moredelim = **[is][\color{blue}]{`}{`},
    commentstyle = \rmfamily\itshape,
    mathescape = true,
    escapebegin = \color[rgb]{0.0, 0.23, 0.0},
    morekeywords = {let,in,as,emit},
    literate={
      {<<}{$\langle$}1
      {>>}{$\rangle$}1
      {lambda}{$\lambda$}1
    }
  }
}

\lstnewenvironment{source}{\csharplstset}{}

\newcommand{\WSL}{{\mathsf{W}}}
\newcommand{\MEL}{{\mathsf{M}}}
\newcommand{\WEL}{{\mathsf{W}}}

\newcommand{\trace}{\mathsf{trace}}
\newcommand{\hist}{\mathsf{hist}}
\newcommand{\inv}{\mathsf{inv}}

\newcommand{\botHist}{\mathsf{hs}} 
\newcommand{\topTrace}{\mathsf{tr}} 

\begin{document}

\setlength{\pdfpageheight}{\paperheight}
\setlength{\pdfpagewidth}{\paperwidth}


\title{Trace Properties from \\ Separation Logic Specifications}

\ifAnonymous

\else

\author{Lars Birkedal\inst{1} \and Thomas Dinsdale-Young\inst{1} \and Guilhem Jaber\inst{2} \and \\ Kasper Svendsen\inst{3} \and Nikos Tzevelekos\inst{4}}
\institute{
  Aarhus University
  (\email{\{birkedal,tyoung\}@cs.au.dk})
  \and
  University Paris Diderot
  (\email{gjaber@pps.univ-paris-diderot.fr})
  \and
  University of Cambridge
  (\email{ksvendsen@cs.au.dk})
  \and
  Queen Mary University of London
  (\email{nikos.tzevelekos@qmul.ac.uk})
}

\authorrunning{L. Birkedal \and T. Dinsdale-Young \and G. Jaber \and K. Svendsen \and N. Tzevelekos}

\fi

\maketitle

\begin{abstract}
We propose a formal approach for relating abstract separation logic library specifications with the trace properties they enforce on interactions between a client and a library.
Separation logic with abstract predicates enforces a resource discipline that constrains when and how calls may be made between a client and a library.
Intuitively, this can enforce a protocol on the interaction trace.
This intuition is broadly used in the separation logic community but has not previously been formalised.
We provide just such a formalisation.
Our approach is based on using wrappers which
instrument library code to induce execution traces for the properties under examination.
By considering a separation logic extended with trace resources, we
prove that when a library satisfies its separation logic specification then its wrapped version satisfies the same 
specification and, moreover, maintains the trace properties as an invariant.
Consequently, any client and library implementation that are correct with respect to the separation logic specification will satisfy the trace properties.
\end{abstract}





\section{Introduction}

Separation logic~\cite{Reynolds:2002:separation-logic,Ishtiaq:2001:bi-asertlog} provides a powerful formalism for specifying an interface between a library and a client in terms of resources.
For example, a client may obtain an ``opened file'' resource by calling an $\tty{open}$ operation of a file library, which it can then use to access the file by calling a $\tty{read}$ operation.
Abstract predicates~\cite{Parkinson:2005:seplog-abstr} crucially hide the implementation details --- the client is not aware of what an ``opened file'' resource actually consists of, which may vary between implementations of the library, but only the functionality for using it that the library provides.
A client that is verified with respect to the abstract specification will behave correctly with any implementation of the library.

In separation logic, functions are specified with preconditions and postconditions.
We can think of the precondition as specifying resources that a client must provide in order to call the function.
The postcondition then specifies the resources that the client receives when the function returns.
For example, a simplified specification for a file library might be the following:
\[
\triple{ \closedp }{\tty{open}()}{ \openp }
\quad
\triple{ \openp }{\tty{close}()}{ \closedp }
\quad
\triple{ \openp }{\tty{read}()}{ \openp }
\]
A client acquires an $\openp$ resource, represented by an abstract predicate, by calling \tty{open}.
With this file resource, the client can call \tty{close} and relinquish the resource, or call \tty{read} and retain the resource.
The rules of separation logic allow us to prove specifications for clients that use the library correctly, such as:\vspace{-1mm}
\[
  \triple{\closedp}{\tty{open}();\, \tty{read}();\, \tty{close}()}{\closedp}\vspace{-1mm}
\]
On the other hand, we cannot prove any useful specifications\footnote{%
  It is always possible to prove a vacuous specification with precondition $\bot$, but
  a useful specification should at least have a satisfiable precondition.
  For complete, closed programs we would typically expect the precondition to be $\emp$.
} for programs that use the library incorrectly, like the following:
  $\tty{open}();\, \tty{close}();\, \tty{read}()$.

Intuitively, separation logic specifications imply properties about the trace of interactions between a library and a client.
For example, the specification for the file library ensures that a file can only be accessed if it has previously been opened and not subsequently closed.
This strongly depends on the fact that the specification is abstract: a client has no way to obtain the $\openp$ resource except by calling the $\tty{open}$ operation.
If the client were able to forge the $\openp$ resource then it could violate the trace property.
While this intuition is broadly used in the separation logic community, it has not previously been formalised.

In this paper we present a formal approach to establishing trace properties from abstract separation logic specifications.
We achieve this by placing a \emph{wrapper} between a client and a library that generates a trace of the interactions between the client and library.
The wrapper has no bearing on the underlying semantics of the program (when traces are ignored) but simply allows us to formally interpret trace properties.
Supposing that the library implements an abstract separation logic specification, we show that the wrapped library also satisfies this specification, 
but moreover maintains the desired trace properties as an invariant.
For the latter step, we recur to a separation logic extended with trace resources, which can be used for instantiating the abstract predicates in the original specification.
In the context of a client that uses the library in accordance with the specification, the trace properties are thus guaranteed to hold. 

Our approach establishes \emph{temporal} trace properties from separation logic specifications that are \emph{not} inherently temporal, independently of implementation details.
%
While ours is not the first approach incorporating temporal reasoning in program logics, and separation logic in particular (q.v.~\S\ref{sec:conclusions}), it is the one 
to derive trace properties for
libraries that already have separation logic specifications.
Previous approaches~\cite{gotsman-rcu,feng,fcsl,tada} achieve temporal reasoning by specifying and verifying
the underlying libraries using trace-oriented specifications; 
the
  novelty of our work is in deriving temporal properties from trace-agnostic specifications.

A motivation for our approach is to \emph{justify} separation logic specifications, by showing that they entail more intuitive trace properties.
Such justification is particularly useful if the specification is intended to formally capture an English-language specification (e.g. the POSIX file system~\cite{POSIXSL}).
While we limit our presentation to the sequential setting and a simple higher-order separation logic, we can show trace properties that include elaborate protocols that are, for instance, parametric on object identifiers or beyond regular languages (q.v~\S\ref{sec:examples}).
%
%

We begin by considering an illustrative example in \S\ref{sec:mot-ex}.
We introduce the programming language (\S\ref{sec:language}) and our separation logic (\S\ref{sec:logic}) before we tackle examples in depth in \S\ref{sec:examples}.
These examples establish that separation logic specifications can imply a variety of trace properties, such as:
\begin{compactitem}
  \item An iterator over a collection should only be used if the collection has not been modified since the iterator was created (\S\ref{sec:iterator-example}).
  \item A higher-order function calls its argument exactly once (\S\ref{sec:wb-example}). 
  \item A traversal on a stack invokes a given method on each value that has been pushed (but not popped) in the order in which they will be popped (\S\ref{stackmap-example}).
  \item Any string received from the user should be sanitised before it is included in an SQL query (\S\ref{sanitize-example}).
\end{compactitem}
The semantics of the logic is presented in \S\ref{sec:semantics}.
Finally, we discuss conclusions and related work in \S\ref{sec:conclusions}.



\section{Motivating Example}\label{sec:mot-ex}

As a motivating example, consider a library that provides a stack with push and pop operations.
In separation logic, these operations can be specified as follows:
\begin{gather*}
  \forall \alpha, a \ldotp \triple{\mathsf{stack}(\alpha) \land a \neq ()}{\texttt{push}(a)}{\mathsf{stack}(a :: \alpha)} \\
  \forall \alpha \ldotp \triple{\mathsf{stack}(\alpha)}{\texttt{pop}()}{r \ldotp (r = () \land \mathsf{stack}(\alpha) \land \alpha = \varepsilon) \lor {} \\ (\exists \alpha' \ldotp \alpha = r :: \alpha' \land \mathsf{stack}(\alpha'))}
\end{gather*}
The \texttt{push} operation simply prepends the given value to the stack, which is represented by the $\mathsf{stack}$ abstract predicate.
The value must be distinct from the unit $()$, which is used to indicate an empty stack.
The \texttt{pop} operation returns $()$ if the stack is empty, and otherwise removes and returns the head value.

We can also specify the stack in terms of trace properties that are satisfied by interactions between a client and a stack library.
For instance, a simple trace property that we might wish to show is this:
\begin{quote}
  Each (non-unit) value returned by an invocation of \texttt{pop} was an argument of a previous invocation of \texttt{push}.
\end{quote}
In order capture trace properties, we define a wrapper that instruments the library operations with code to emit trace events as:
\begin{align*}
  \texttt{push}' &\eqdef \lambda v \ldotp \texttt{push}(v); \TEMPemit \langle \texttt{push}, v \rangle &&&
  \texttt{pop}' &\eqdef \lambda \_ \ldotp \TEMPletin{r = \texttt{pop}()}{\TEMPemit \langle \texttt{pop}, r \rangle; r}
\end{align*}
We can then define a language of traces that captures our desired invariant:\vspace{-1mm}
\[
  \mathcal{L} = \Setb{t}{\forall i, v \ldotp t[i] = \langle \texttt{pop}, v \rangle \land v \neq () \implies \exists j \ldotp j < i \land t[j] = \langle \texttt{push}, v \rangle}
\]

Our separation logic includes two resources that allow us to reason about code that emits trace events: $\trace(t)$ expresses that $t$ is the current trace; and $\inv(I)$ expresses that the trace invariant is $I$.
The proof rule for $\TEMPemit$ updates the $\trace$ resource, while requiring that the new trace satisfies the invariant.
Using these, we can define a `wrapped' version of the abstract $\mathsf{stack}$ predicate so that the wrapped operations will satisfy the same abstract specification, but also enforce the trace invariant:
\[
  \mathsf{stack}'(\alpha) \eqdef \mathsf{stack}(\alpha) * \inv(\mathcal{L}) *
    \exists t \ldotp \trace(t) \land t \in \mathcal{L} \land \forall a \in \alpha \ldotp \exists i \ldotp t[i] = \langle \texttt{push}, a \rangle
\]
The proof of the wrapped push operation proceeds as follows:
\begin{source}
$\{ \mathsf{stack}'(\alpha) \land v \neq () \}$
$\{ \mathsf{stack}(\alpha) \land v \neq () * \inv(\mathcal{L}) * \exists t \ldotp \trace(t) \land t \in \mathcal{L} \land \forall a \in \alpha \ldotp \exists i \ldotp t[i] = \langle \texttt{push}, a \rangle \}$
  push(v);
$\{ \mathsf{stack}(v :: \alpha) * \inv(\mathcal{L}) * \exists t \ldotp \trace(t) \land t \in \mathcal{L} \land \forall a \in \alpha \ldotp \exists i \ldotp t[i] = \langle \texttt{push}, a \rangle \}$
$\{ \mathsf{stack}(v :: \alpha) * \inv(\mathcal{L}) * \exists t \ldotp \trace(t) \land (t \cdot \langle \texttt{push}, v \rangle) \in \mathcal{L} \land {} $
    $\forall a \in (v :: \alpha) \ldotp \exists i \ldotp (t \cdot \langle \texttt{push}, v \rangle)[i] = \langle \texttt{push}, a \rangle \}$
  emit<<push,v>>;
$\{ \mathsf{stack}(v :: \alpha) * \inv(\mathcal{L}) * \exists t \ldotp \trace(t) \land t \in \mathcal{L} \land \forall a \in \alpha \ldotp \exists i \ldotp t[i] = \langle \texttt{push}, a \rangle \}$
$\{ \mathsf{stack}'(v :: \alpha) \}$
\end{source}
The wrapped pop operation can be verified similarly.
We can thus conclude that the stack indeed satisfies the desired trace property.

The above example demonstrates our technique on a first-order library (the library cannot make call-backs to the client), but it also applies in a higher-order setting.
For instance, consider extending the stack with a \texttt{foreach} operation that traverses the stack and calls a client-supplied function on each element in order.
In separation logic, this operation can be specified as:
\[
  \forall \alpha,\! f,\! I \ldotp\! \triple{\mathsf{stack}(\alpha) * I(\varepsilon) * 
\forall \beta, a \ldotp\\
 \triple{I(\beta)}{f(a)}{I(a :: \beta)}}{\!\texttt{foreach}(f)\!}{\mathsf{stack}(\alpha) * {}\\ I(\mathit{rev}(\alpha))}
\]
This specification is subtle.
The \texttt{foreach} operation takes a function $f$ that is specified (using a nested triple) with an invariant $I(\beta)$ whose parameter records the list of values that $f$ has so far been called on.
The operation is given the predicate $I(\varepsilon)$ initially, and it returns with $I(\mathit{rev}(\alpha))$, where $\mathit{rev}(\alpha)$ is the reversal of the list $\alpha$.
Note that while the predicate $\mathsf{stack}$ is abstract to the client (the library determines the interpretation), the predicate $I$ is abstract to the library (the client determines the interpretation).
This means that for the library to obtain the $I(\mathit{rev}(\alpha))$ predicate for the postcondition from the $I(\varepsilon)$ given in the precondition, it must call $f$ on each element of $\alpha$ in order.
Moreover, it cannot make any further calls to $f$, since separation logic treats the predicate $I(\beta)$ as a resource, which must be given up by the library each time it calls $f$, and which the library has no means of duplicating since in separation logic $A \implies A * A$ does not hold in general.

Defining a wrapper for higher-order libraries is more complex than in the first-order case, since we wish the trace to capture all interactions between the client and the library, including call-backs between them.
The wrapper must therefore emit events at the call and return of library functions, as well as functions that are passed as arguments to library functions.
We can then specify a number of properties that traces generated by client-library interactions will have:
\begin{compactitem}
  \item The trace of \texttt{push} and \texttt{pop} operations obeys the stack discipline.
  \item Between invocation and return, $\texttt{foreach}(f)$ calls $f$ on each element of the stack in order, with no further calls.
  \item The invocations of \texttt{push}, \texttt{pop} and $f$ (the argument of a call of $\texttt{foreach}$) are atomic --- there are no further events between the call and return.
\end{compactitem}
While the first property can be seen as a straightforward consequence of the \texttt{push} and \texttt{pop} specifications, the others are more subtle.
In particular, they depend on the \texttt{foreach} specification's parametricity in $I$, which prevents the library from using its argument in an arbitrary fashion.
For particular instantiations of $I$, (e.g. $I(\beta) = \true$), \texttt{foreach} could call its argument an arbitrary number of times, or even store it and call it from future invocations of \texttt{pop}.
However, parametricity ensures that it will behave the same independent of how $I$ is instantiated, and so it cannot do that since
$I$ can be instantiated so as to enforce trace properties.

This example illustrates several important aspects of our approach.
Firstly, we support higher-order functions, such as \texttt{foreach}.
We also deal with expressive trace properties: the language of traces is not context-free, since the stack may be traversed multiple times.
Moreover, the connection between the separation logic specification and the trace properties that follow from it is subtle.

In \S\ref{sec:examples} we revisit this example, among others, in detail.
Before doing so, we present the formal setting of our approach.



\section{Programming Language}
\label{sec:language}

We define the programming language which will be the object of this study. To keep the setting general, the language is an untyped call-by-value $\lambda$-calculus with references, named functions, pairs,
integers and primitives for arithmetic, reference usage and conditionals.
Our language additionally features an $\emit$ primitive, which is used to output values as trace events.
We will formulate properties in terms of the event traces produced by program evaluation.
The syntactic classes of \emph{values} ($\VAL$), \emph{expressions} ($\EXP$) and \emph{evaluation contexts} ($\CON$) are given respectively as follows.
\[\begin{array}{@{}r@{\ }c@{\ }l}
 u,v & ::=&  \unit \sep n \sep x \sep l \sep f \sep \pair{u}{v}\\
 e & ::= &  v \sep \lambda x.e \sep e \ e' \sep \ifte{e_1}{e_2}{e_3}  \sep 
    \pair{e}{e'}\sep \proj{i}(e) \sep \refer e \sep {!}e \sep e \,{\sf op}\, e'  \sep \emit\, v  \\
 K & ::= & \hole \sep K e \sep v K \sep K \,{\sf op}\, e \sep v \,{\sf op}\, K \sep \ifte{K}{e}{e'} 
    \sep \pair{K}{e} 
    \sep \pair{v}{K} \sep \proj{i}(K)\sep \refer K \sep {!}K
\end{array}\]
Above we let $ i \in \{1,2\}, n \in \mathbb{Z}, f \in \FUN$ and $l \in \LOC$,
where $\FUN$ and $\LOC$ are countable sets of function and location identifiers respectively. 
Moreover, $x$ ranges over the set of variables $\VAR$, and ${\sf op}\in\{=,:=, \dotsc\}$ ranges over a set of binary operators which includes equality test, assignment and arithmetic operators.
We use $\LAM$ for the set of $\lambda$-abstraction expressions.
Note that $\lambda$-abstractions are not values; functional values are represented by function identifiers.

Our operational semantics tracks $\lambda$-abstractions in an environment $\gamma:\FUN\rightharpoonup_{\rm fin}\LAM$, where all encountered functions are named and stored.
This allows us to track function usage by emitting trace events that refer to these function names.
In addition, we can reduce expressions that use external functions: in such a case, an expression $e$ can contain the names of the external functions in its code, and $\gamma$ would provide their bodies. 
The semantics draws from the open-term trace semantics used e.g.\ in~\cite{Laird:2007:fulabstr-tracesem}, the main difference being that we explicitly control  events generated
by the reduction via the $\emit$ primitive (and also indiscriminately name $\lambda$-abstractions).

Expressions are evaluated inside states $(h,\gamma) \in \HEAP\times\FENV$ comprising a function environment $\gamma$ and a heap $h$.
A heap is a finite map from locations to values and a function environment is a finite map from function identifiers to $\lambda$-abstractions.
For any domain/codomain set pair $X,Y$, map $g:X\rightharpoonup Y$ and $(x,y)\in X\times Y$, we let 
$g[x\mapsto y]=\{(x,y)\}\cup\{\,(z,g(z))\mid z\in\dom{g}\setminus\{x\}\,\}$
regardless of whether $x\in\dom{g}$ or not.
The evaluation rules are:
\[\begin{array}{@{}r@{\ }c@{\ }l@{\quad}r}
\config{K[\lambda x.e],(h,\gamma)} & \rightarrow &   \config{K[f],(h,\gamma [f\mapsto \lambda x \ldotp e])} &  (f \notin \dom{\gamma})\\
\config{K[f \, v],(h,\gamma)} &  \rightarrow &  \config{K[e\{v/x\}],(h,\gamma)} & (\gamma(f) = \lambda x \ldotp e)\\
\config{K[\proj{i}(\pair{v_1}{v_2})],(h,\gamma)} & \rightarrow &  \config{K[v_i],(h,\gamma)} & (i \in \{1,2\}) \\
\config{K[\ifte{n}{e_1}{e_2}],(h,\gamma)} & \rightarrow & \config{K[e_1],(h,\gamma)} & (\text{if } n > 0)\\
\config{K[\ifte{0}{e_1}{e_2}],(h,\gamma)} & \rightarrow & \config{K[e_2],(h,\gamma)} & \\ 
\config{K[\refer v],(h,\gamma)} & \rightarrow & \config{K[l],(h[l \mapsto v],\gamma)} & (l\notin \dom{h})\\
\config{K[!l],(h,\gamma)} & \rightarrow & \config{K[h(l)],(h,\gamma)} & (l\in \dom{h})\\
\config{K[l:=v],(h,\gamma)} & \rightarrow & \config{K[\unit],(h[l\mapsto v],\gamma)} & (l\in \dom{h})\\
\config{K[\emit \, v],(h,\gamma)} & \xrightarrow{v} & \config{K[()],(h,\gamma)} 
\end{array}\]
The above defines a labelled transition system with labels from $\VAL\cup\{\epsilon\}$ and, by taking its reflexive transitive closure ($\to^*$), we obtain labels from $\VAL^*$, which we shall call \emph{traces}.

We next relate the behaviour of each expression with that of its counterpart where all emits have been omitted.
We write $\erase{e}$ for the expression obtained from $e$ by replacing all occurrences of $\emit\, v$ by $\unit$, and extend this notation to evaluation contexts and functional environments in the expected manner.

\begin{theorem}\label{thm:erasure-emit}
For all $e,e',h,h',\gamma,\gamma'$,
if $\config{\erase{e},(h,\erase{\gamma})} \rightarrow^* \config{e',(h',\gamma')}$ 
then there exists a trace $t$, an expression $e''$ and an environment $\gamma''$ such that
$e' = \erase{e''}, \gamma' = \erase{\gamma''}$ and
$\config{e,(h,\gamma)} \xrightarrow{t}{\!\!}^*\ \config{e'',(h',\gamma'')}$.
\end{theorem}

This theorem justifies the fact that instrumenting terms with emits does not change their semantics.
Its proof is in Appendix~\ref{sec:erase-emit}.



\section{Logic}
\label{sec:logic}

In this section we introduce our logic and prove some meta-theoretic results. 
The logic is a standard higher-order separation logic
for the $\lambda$-calculus introduced in \S\ref{sec:language}, extended with new
primitives for reasoning about trace properties and programs that emit traces. 

The logic consists of an assertion logic for reasoning about machine states
and a specification logic for reasoning about the behaviour of programs. The
assertion and specification logics are constructed over the same simply-typed
term language. The types of this simply-typed term language are given below:
\begin{align*}
  \sigma, \tau\, &::=\
    \ty{1} \mid \ty{Bool} \mid \ty{Nat} \mid \tau \to \sigma \mid \tau \times \sigma 
    \mid \ty{seq}~\tau \mid \ty{Prop} \mid \ty{Spec} \mid \ty{Val} \mid \ty{Exp} \mid \ty{Loc}
\end{align*}
The types include standard base types (unit, Booleans, and natural numbers) and type formers for functions, products and finite sequences.
Additionally, there is a type for resource assertions (\ty{Prop}), a type for specifications (\ty{Spec}),
and types for values, expressions and locations from the programming language
(\ty{Val}, \ty{Exp} and \ty{Loc}, respectively). The type $\ty{seq}~\tau$ is the
type of finite sequences of elements of type $\tau$; we will use this type to
represent traces.

The terms of the language are generated by the following grammar. The language
includes a simply-typed $\lambda$-calculus with pairs, a higher-order logic with
equality, 
primitive operators on finite sequences ($\nil,\cons,\fold$),
primitive separation logic resources ($\emp$, $*$, $\wand$, $\mapsto$), 
Hoare triples and embeddings back and forth between the specification and assertion 
logic ($\valid$ and $\spec$) and three trace primitives ($\trace$, $\hist$ and $\inv$). 
\begin{align*}
& P,Q,N,M ::= x \mid \lambda x {\,:\,} \tau \ldotp M \mid M N \mid () \mid (M,N) \mid \pi_i(M) \mid \bot \mid \top \mid P \lor Q \mid P \land Q \\
&\qquad\!\! 
 \mid P {\,\implies\,} Q \mid \forall x {\,:\,} \tau \ldotp P \mid \exists x {\,:\,} \tau \ldotp P \mid M =_\tau N   \mid \true \mid \false 
\mid \Tif~M~\Tthen~N_1~\Telse~N_1
\\
&\qquad\!\! 
\mid \valid(P) \mid \emp \mid P * Q \mid P \wand Q \mid M \mapsto N \mid \triple{P}{e}{Q}\mid \spec(M) \mid \nil_\tau \\
&\qquad\!\! 
\mid \cons(M, N) \mid \fold(M, N_1, N_2) 
 \mid \trace(M) \mid \hist(M) \mid \inv(M)
\end{align*}
The typing rules for the $\lambda$-calculus part are standard and have been 
omitted. Figure \ref{fig:logic-typing} includes an excerpt of some of the
more interesting typing rules.

\begin{figure}[t]
\begin{gather*}
  \inferrule[TPointsTo]{
    \Gamma \vdash M : \ty{Loc} \\
    \Gamma \vdash N : \ty{Val}
  }{
    \Gamma \vdash M \mapsto N : \ty{Prop}
  }\quad
  \inferrule[THoare]{
    \Gamma \vdash P : \ty{Prop} \\
    \Gamma \vdash M : \ty{Exp} \\
    \Gamma \vdash Q : \ty{Val} \rightarrow \ty{Prop}
  }{
    \Gamma \vdash \triple{P}{M}{Q} : \ty{Spec}
  }
\\
  \inferrule[TValid]{
    \Gamma \vdash P : \ty{Prop}
  }{
    \Gamma \vdash \valid(P) : \ty{Spec}
  }
\quad
  \inferrule[TSpec]{
    \Gamma \vdash M : \ty{Spec}
  }{
    \Gamma \vdash \spec(M) : \ty{Prop}
  }
\quad
  \inferrule[TFold]{
    \Gamma \vdash M : \ty{seq}~\tau \\\hspace{-4mm}
    \Gamma \vdash N_1 : \sigma \\\hspace{-4mm}
    \Gamma \vdash N_2 : \sigma {\times} \tau \rightarrow \sigma
  }{
    \Gamma \vdash \fold(M, N_1, N_2) : \sigma
  }\\[-9mm]
\end{gather*}
\caption{Excerpt of typing rules.}
\label{fig:logic-typing}\vspace{-2mm}
\end{figure}

The typing judgement has the form 
$\Gamma \vd M : \tau$ where $\Gamma$ is a term context associating types
with variables. Note that the typing rule for Hoare triples (\textsc{THoare})
takes an assertion $P$ as a precondition and a predicate $Q$ as a postcondition
to allow the postcondition to refer to the return value of the computation.
Since we are reasoning about a $\lambda$-calculus
(so that variables are immutable), we do not distinguish program 
and logical variables. The expression $e$ in a Hoare triple is thus typed
in the same context as the Hoare triple, allowing us to refer to
logical variables as program variables. 

\begin{figure}[t]
\begin{mathpar}
  \inferrule[Hyp]{
    \Gamma \mid \Theta, S \vdash S
  }{}
\and 
  \inferrule[Ret]{
   \Gamma, v : \ty{Val} \mid \Theta \vdash \{ \top \}~v~\{ r.~r = v \}
  }{}
\and
  \inferrule*[right=Valid]{
    \Gamma \mid \Theta \mid - \vdash P
  }{
    \Gamma \mid \Theta \vdash \valid(P)
  }
\and
  \inferrule*[right=SpecOut]{
    \Gamma \mid \Theta, S \vdash \triple{P}{e}{R}
  }{
    \Gamma \mid \Theta \vdash \triple{P * \spec(S)}{e}{R}
  }\vspace{-1mm}
\and
  \inferrule*[right=Frame]{
    \Gamma \mid \Theta \vdash \triple{P}{e}{Q}
  }{
    \Gamma \mid \Theta \vdash \triple{P * R}{e}{r.~Q(r) * R}
  }
\and
  \inferrule*[right=Bind]{
    \Gamma \mid \Theta \vdash \triple{P}{e}{x \ldotp Q} \\
    \Gamma, x : \ty{Val} \mid \Theta \vdash \triple{Q}{K{[x]}}{r \ldotp R} \\
    x \not\in FV(\Theta)
  }{
    \Gamma \mid \Theta \vdash \triple{P}{ K{[e]} }{ r \ldotp \exists x : \ty{Val} \ldotp R}
  }
\and
  \inferrule*[right=Csq]{
    \Gamma \mid \Theta \mid P_1 \vd P_2 \\
    \Gamma \mid \Theta \vdash \triple{P_2}{e}{Q_2} \\
    \Gamma, x : \ty{Val} \mid \Theta \mid Q_2(x) \vd Q_1(x)
  }{
    \Gamma \mid \Theta \vdash \triple{P_1}{e}{Q_1}
  }
\and
  \inferrule*[Right=Abs]{
    \Gamma, x : \ty{Val} \mid \Theta \vdash \triple{P}{e}{Q} \\
    x \not\in FV(\Theta)
  }{
    \Gamma \mid \Theta \vdash \triple{\emp}{\lambda x \ldotp e}{r \ldotp \forall x : \ty{Val} \ldotp \spec(\triple{P}{r x}{Q})}
  }
\and
\and
\inferrule[Write]{
    \Gamma, v, l : \ty{Val} \mid \Theta \vdash \triple{l \mapsto \_}{l := v}{r.~l \mapsto v * r = ()}
  }{}
\and
  \inferrule[PInvDupl]{
    \Gamma \mid \Theta \mid \inv(I) \vd \inv(I) * \inv(I)
  }{}
\and
\begin{gathered}
\inferrule[Read]{
    \Gamma, v, l : \ty{Val} \mid \Theta \vdash \triple{l \mapsto v}{!l}{r.~l \mapsto v * r = v}
  }{}
\quad
  \inferrule[Alloc]{
    \Gamma, v : \ty{Val} \mid \Theta \vdash \triple{\emp}{\refer v}{r \ldotp r \mapsto v }
  }{}
\end{gathered}
\and
  \inferrule[PHistDupl]{
    \Gamma \mid \Theta \mid \hist(t) \vd \hist(t) * \hist(t)
  }{}
\and
  \inferrule[PUseHist]{
    \Gamma \mid \Theta \mid \trace(t_1) * \hist(t_2) \vd \trace(t_1) * t_2 \leqprefix t_1
  }{}
\and
\begin{gathered}
  \inferrule[PAllocHist]{
    \Gamma \mid \Theta \mid \trace(t) \vd \trace(t) * \hist(t)
  }{}
\quad
  \inferrule*[right=Emit]{
    \Gamma \mid \Theta \vdash I(t \cdot v)
  }{
    \Gamma \mid \Theta \vdash \triple{ \trace(t) * \inv(I) }{\emit\, v}{ \trace(t \cdot v) }
  }
\end{gathered}
\vspace{-4mm}
\end{mathpar}
\caption{Selected specification and assertion entailments.}
\label{fig:specrules}\label{fig:elrules}\vspace{-2mm}
\end{figure}

Judgements of the assertion logic have the form $\Gamma \mid \Theta \mid P \vd Q$, where $\Gamma$ is again a term context, associating types to variables, $\Theta$ is context of assumed specifications, and $P$ and $Q$ are resource assertions.
The judgement should be interpreted as: under the assumption the specifications $\Theta$ hold, the resource assertion $P$ entails the resource assertion $Q$.
The assertion logic consists of the usual entailment rules for higher-order separation logic~\cite{Biering:2007:bi-hyp-hosl} with the following additions, 
\[
  \inferrule[PSpec]{
    \Gamma \mid \Theta, S \mid P \vd \spec(S)
  }{}
\qquad
  \inferrule[PValid]{
    \Gamma \mid \Theta, \valid(Q) \mid P \vd Q
  }{}
\]
which allow us to use the specification context in propositional entailments.

Judgements of the specification logic have the form $\Gamma \mid \Theta \vd S$, where again $\Gamma$ and $\Theta$ are term and specification contexts respectively, and $S$ is a specification
The specification logic consists of the usual entailment rules for higher-order logic, with the addition of the rules in Figure~\ref{fig:specrules}.
As a notational convention, we drop the $\lambda$ in the postcondition: $\triple{P}{e}{r \ldotp Q}$ means $\triple{P}{e}{\lambda r : \ty{Val} \ldotp Q}$, and 
$\triple{P}{e}{Q}$ stands for $\triple{P}{e}{\lambda \_ : \ty{Val} \ldotp Q}$.

\begin{example}
Consider the specification $\Phi(P_\mathit{init},\mathit{bracket})$ defined as follows:
\begin{align*}
&\exists \mathsf{inv} : \ty{Prop} \ldotp 
    \valid(P_\mathit{init} \implies \mathsf{inv}) \land {} 
     \forall P, Q : \ty{Prop} \ldotp \forall f : \ty{Val} \ldotp \\
    & \hspace{1em} \triple{\mathsf{inv} * P * \spec(\triple{P}{f()}{Q})}{ \mathit{bracket}(f) }{\mathsf{inv} * Q}
\end{align*}
This specification describes a module which requires some initial resource $P_\mathit{init}$ to establish its invariant $\mathsf{inv}$.
The assertion $P_\mathit{init} \implies \mathsf{inv}$ permits the invariant to be constructed; since this is a resource assertion, it is wrapped in $\valid$ to produce the specification which asserts that the implication holds for all resources.
The module provides one function, $\mathit{bracket}$, that is specified by the Hoare triple.
The specification states that, when $\mathit{bracket}$ is applied to a function $f$ that takes precondition $P$ to postcondition $Q$, then it will behave similarly in the presence of the invariant $\mathsf{inv}$.
The condition on the argument $f$ is specified with a nested triple in the precondition; since a triple is a $\ty{Spec}$ and the precondition must be a $\ty{Prop}$, the triple must be wrapped by $\spec$. 

As a trivial implementation, we can prove\, $\text{-}\,|\, \text{-} \vdash \Phi(\emp, \lambda f \ldotp f ())$.
A client specification $S$ may be proved against an abstract module by proving:
\[
P_\mathit{init} : \ty{Prop}, \mathit{bracket} : \ty{Val} \mid \Phi(P_\mathit{init},\mathit{bracket}) \vd S
\]
This can then be composed with the module implementation.
\qed
\end{example}

Note that the specification logic lacks an application rule for
applying an argument to a function. Instead, the abstraction rule
returns a specification for the given function, applied to an arbitrary
argument $x$. To use this specification, we instantiate
the specification with the actual function argument and pull
out the nested triple to the context using the \textsc{SpecOut}
rule.

\paragraph{Trace primitives.}
The logic includes three basic assertions, $\trace(t)$, $\hist(t)$
and $\inv(I)$, for reasoning about traces. The typing rules for these
primitive trace assertions are given below. We use $\ty{Trace}$ as 
shorthand for $\ty{seq}\,\ty{Val}$.
\[
  \inferrule*[right=TTrace]{
    \Gamma \vdash t : \ty{Trace}
  }{
    \Gamma \vdash \trace(t) : \ty{Prop}
  }
\quad
  \inferrule*[right=THist]{
    \Gamma \vdash t : \ty{Trace} 
  }{
    \Gamma \vdash \hist(t) : \ty{Prop}
  }
\quad
  \inferrule*[right=TInv]{
    \Gamma \vdash I : \ty{Trace} \rightarrow \ty{Bool}
  }{
    \Gamma \vdash \inv(I) : \ty{Prop}
  }
\]

The trace resource, 
$\trace(t)$, expresses that the trace of events
emitted so far is exactly $t$ and asserts exclusive right to emit further trace
events. The trace $t$ is represented as a finite sequence of values. Since $\trace(t)$
asserts exclusive right to emit events, it cannot be duplicated. The $\trace$ resource thus allows a single owner to reason precisely
about the current trace. 

The $\trace(t)$ resource suffices for examples where only a single resource needs
to refer to the current trace. To improve expresiveness, we use the $\hist(t)$
resource to reason about prefixes of the current trace. 
The resource $\hist(t)$ asserts that the trace $t$ is a prefix of the trace of events
emitted so far. Note that this property is preserved by emission of new events:
if $t$ is a prefix of the current trace then $t$ is also a prefix of any extension
of the current trace.
This resource is duplicable, and given the $\trace(t)$ resource we can construct a $\hist(t)$.
Moreover, if we have a $\trace(t_1)$ resource and a $\hist(t_2)$ resource, we can conclude that $t_2$ is a prefix of $t_1$.
These properties are captured in the axioms \textsc{PHistDupl}, \textsc{PAllocHist} and \textsc{PUseHist} rules of the extended logic, given in Figure~\ref{fig:elrules}.

For several of our examples we require that all traces generated
belong to a restricted language of traces. We express this formally using the invariant
resource, $\inv(I)$, which defines a trace invariant that everyone must obey. Here $I$ is 
a set of traces and $\inv(I)$ asserts that the current trace invariant is given by $I$.
Since it specifies an invariant, the $\inv(I)$ resource is duplicable (axiom \textsc{PInvDupl} in Figure~\ref{fig:elrules}).
A triple $\triple{P}{e}{Q}$ expresses
that for every initial state satisfying $P$ if $e$ executes to a terminal
state then the terminal state satisfies $Q$ and the current trace at every intermediate
state (including the initial and terminal state) satisfies the trace invariant.

To emit an event $v$ we thus require ownership of the trace resource $\trace(t)$ and that
the current trace after emitting the event satisfies the trace invariant $I$.
The proof rule \textsc{Emit} in Figure~\ref{fig:elrules} captures this.
Note that the $t \cdot v \in I$ assumption in the \textsc{Emit} rule is not strictly 
part of the syntax of our term language, but may be defined using $\fold$. The same is 
true of $\leqprefix$ and other operations on finite sequences, such as concatenation, 
length of a sequence and a subsequence operation. In subsequent sections we will use
these definable operations without further mention. We will also occasionally
need inductively defined predicates on traces. Since we are working in a higher-order
logic, such predicates are definable within the logic using the usual impredicative
Knaster-Tarski definition of least-fixed points of monotone operators. 





\section{Proving Trace Properties}\label{sec:examples}

In this section we show how to derive trace properties about client-library interactions from the library's specification.
We demonstrate our approach through
a series of increasingly complex examples, starting with the basic file
library example from the Introduction. 

The basic idea is to prove that for any library implementation satisfying
the abstract library specification, the wrapped library implementation
satisfies the same abstract specification and moreover the traces generated
by the wrapping satisfy a given invariant. This is achieved by reinterpreting
the abstract representation predicates of the specifications to additionally 
relate the abstract state with the current trace using trace assertions. 
Client programs verified against the abstract library interface can thus be 
linked with the wrapped library implementation to conclude that the traces
generated by the wrapping satisfy the given invariant. 

Theorem \ref{thm:sl-trace-prop} below formalises this idea. Here
$I_\mathit{init}$ is an initialiser operation to initialize the
internal state $P_0$ of the library. Theorem \ref{thm:sl-trace-prop}
allows us to prove that traces generated by running a client $e$
linked with a library implementation $I_{ops}$, after running the
initialiser, belong to a given language $\mathcal{L}$.
\begin{theorem}\label{thm:sl-trace-prop}
Given any specification $\Phi$, resource $P_0$, initialiser $I_\mathit{init}$,
library implementation $I_\mathit{ops}$, machine states  $s,s'$, 
trace $t$, clients $e,e'$ and language $\mathcal{L}$,
if the following conditions hold then $t \in \mathcal{L}$:
\begin{itemize}
\item
$- \vd \Phi : \ty{Prop} \times  \ty{Exp} \rightarrow \ty{Spec}$ and $- \vd P_{0} : \ty{Prop}$ 
\item
  $- \mid - \vd \forall \mathit{ops}, P.~\Phi(P,\mathit{ops}) \implies \{ P \}~e~\{ \top \}$
\item
  $- \mid - \vd \{ \top \}~I_\mathit{init}~\{ P_{0} \}$
\item
  $- \mid - \vd \Phi(P_{0} * \trace(\varepsilon) * \inv(\mathcal{L}),I_{ops})$
\item
  $(I_\mathit{init}; e[I_\mathit{ops}/\mathit{ops}]), s \stackrel{t}{\rightarrow}{\!}^*\,\, e',s' \not\rightarrow$
\end{itemize}
\end{theorem}
\begin{proof}
Follows from soundness of the logic (Lemma \ref{lem:el-multistep-exec}
in \S\ref{sec:semantics}).
\qed
\end{proof}
The above theorem requires that we provide a library implementation
that satisfies the abstract specification $\Phi$ and generates
traces in the language $\mathcal{L}$. 
To use this theorem to derive a trace
property from the separation logic specification,\linebreak 
the idea is to define a suitable
wrapper function $\wrap$ for the library in question\linebreak 
and prove that if a library
implementation $M$ satisfies $\Phi$ then so does the wrapped\!\linebreak 
version $\wrap(M)$
\emph{and}, additionally, the wrapped version generates traces in the $\mathcal{L}$\!\linebreak 
language:
  $P_0 : \ty{Prop}, \mathit{ops} : \ty{Exp} \mid \Phi(P_0, \mathit{ops}) \vdel
                                           \Phi(P_0 * \trace(\varepsilon) * \inv(\mathcal{L}), \wrap(\mathit{ops}))$.

\subsection{File Library}\label{ex:file-lib}

To illustrate the idea, we begin by recalling the file library example from the Introduction.
We prove that the associated separation logic specification enforces that clients
verified against the specification only close and read when the file is open. 
To capture this property, we define a wrapper function $\wrapfile$ that 
instruments an implementation of the file library to emit events about calls to
\tty{open}, \tty{close} and \tty{read}. Formally, we take a file library to be
a triple consisting of an \tty{open}, a \tty{close} and a
\tty{read} function.
\begin{multline*}
\wrapfile \eqdef \lambda (\mathit{open}, \mathit{close}, \mathit{read}) \ldotp \\
  \big(  \lambda \_ \ldotp \mathit{open}();\TEMPemit\,\tty{open}, \ \lambda \_ \ldotp \mathit{close}();\TEMPemit\,\tty{close}, \ 
  \lambda \_ \ldotp \mathit{read}();\TEMPemit\,\tty{read} \big)
\end{multline*}
We can now formalize the protocol as constraints on the traces generated by linking
an instrumented file library with a client. In particular, we require that traces
belong to the language $\Lfile$ of all strings $t \in \Sigma^*$ such that $t$ is a
valid file trace, $\validfilep(t)$, where
$\Sigma = \{ \tty{open}, \tty{close}, \tty{read} \}$ and $\validfilep$
is defined as: 
\begin{align*}
\validfilep(t) &\eqdef
  \forall n \ldotp 
    t[n] = \tty{read} \lor t[n] = \tty{close} \implies \isopenp(t, n) \\
\noclosep(t, n, m) &\eqdef
  \forall k \ldotp n < k < m \implies t[k] \neq \tty{close} \\
\isopenp(t, n) &\eqdef
  \exists m < n \ldotp t[m] = \tty{open} \land \noclosep(t, m, n)
\end{align*}
The valid file trace predicate, $\validfilep(t)$, expresses that
the trace $t$ only contains read and close events when the file is open. 
Recall the SL specification of the file library from the Introduction, here
written more formally.\vspace{-.5mm}
\begin{multline*}
\Phifile(P_{0}, (\mathit{open}, \mathit{close}, \mathit{read})) \eqdef \exists \openp, \closedp\! :\! \ty{Prop} \ldotp
  \valid(P_{0} {\implies} \closedp) \land {} \\
  \triple{\closedp}{\mathit{open}()}{\openp} \land 
    \triple{\openp}{\mathit{close}()}{\closedp} \land 
 \triple{\openp}{\mathit{read}()}{\openp}
\end{multline*}
To prove the specification enforces the intended protocol, we proceed by proving 
that the wrapping preserves satisfaction of the specification and generates traces in $\Lfile$.
\begin{lemma}
$P_0, \mathit{ops} \mid \Phifile(P_0, \mathit{ops}) \vdel
\Phifile(P_0 * \trace(\varepsilon) * \inv(\Lfile), \wrapfile(\mathit{ops}))$.
\end{lemma}
\begin{proof}[sketch]
We first need to define new wrapped versions of the abstract representation predicates, which
relate the abstract resources to the current trace. The idea is that the wrapped $\openp$ resource,
$\openpw$, should express that the current trace $t$ is in $\Lfile$ and that the trace $t$ is open.
Likewise, the wrapped $\closedp$ resource, $\closedpw$ should simply express that the current trace
is in $\Lfile$. 
\begin{align*}
\openpw &\eqdef \openp * \exists t \in \Lfile \ldotp \trace(t) * \inv(\Lfile) * \isopenp(t, |t|+1) \\[0mm]
\closedpw &\eqdef \closedp * \exists t \in \Lfile \ldotp \trace(t) * \inv(\Lfile)
\end{align*}
We need to prove that $P_0 * \trace(\varepsilon) * \inv(\Lfile) \implies \closedpw$ assuming
\hbox{$P_0 \implies \closedp$}, which follows trivially from the definition of $\closedpw$. It remains
to prove that the wrapped methods satisfy their specifications. Below we give a proof outline for
the wrapped \tty{open} method.
\begin{source}
$\{ \closedpw \}$
$\{ \closedp * \exists t \in \Lfile \ldotp \trace(t) * \inv(\Lfile) \}$
open();
$\{ \openp * \exists t \in \Lfile \ldotp \trace(t) * \inv(\Lfile) \}$
emit open;
$\{ \openp * \exists t \in \Lfile \ldotp \trace(t) * \inv(\Lfile) * \isopenp(t, |t|+1) \}$
$\{ \openpw \}$
\end{source}
Here we use the assumed specification of the underlying \tty{open} method 
to verify the call to open and we used the following property of $\Lfile$ and $\isopenp$ to
prove that emitting \tty{open} would result in a trace in $\Lfile$ that was open.
\begin{align*}
\forall t \in \Lfile  \ldotp t \cdot \tty{open} \in \Lfile \land \isopenp(t \cdot \tty{open}, |t \cdot \tty{open}|+1)
\end{align*}
The proof outlines for the \tty{close} and \tty{read} operations are similar, but use
the following property to justify the emits.
\[
  \forall t  \in \Lfile\ldotp \isopenp(t, |t|+1)   \implies
 t \cdot \tty{close} \in \Lfile \land t \cdot \tty{read} \in \Lfile \qedhere
\]
\end{proof}

\subsection{Iterators on Collections}\label{sec:iterator-example}

We consider a collections library that provides methods for modifying a collection as well as iterating over it.
To ensure a well-defined semantics for iterators, we require the following property:
\begin{quote}\em
An iterator over a collection should only be used if the underlying
collection has not been destructively modified since the iterator was created.
\end{quote}
This is a trace property of the interaction between the collections library
and clients.
To capture it formally, we first define a suitable wrapper for the library that produces appropriate trace events.
We take a
collections library to be a tuple consisting of five operations: \tty{size},
\tty{add}, \tty{remove}, \tty{iterator} and \tty{next}. The \tty{size}
operation is non-destructive
and returns the size of the given collection. The \tty{add} and
\tty{remove} operations destructively modify the given collection by adding or
removing an element from the collection. Finally, \tty{iterator} returns a new iterator
for the collection, while \tty{next} returns the next element of a given iterator.

The instrumentation is fairly straightforward and simply emits a suitable event
indicating the operation called and the argument and/or return value of the given
operation, when relevant:
\begin{multline*}
\wrapcoll \eqdef \lambda  (\mathit{size},  \mathit{add}, \mathit{remove}, \mathit{iter}, \mathit{next}) \ldotp \\
\left( \begin{array}{@{}l@{}}
    \lambda y\ldotp\TEMPletin{r = \mathit{size}(y)}{\TEMPemit\,\tty{size};r}, \ 
    \lambda y\ldotp\mathit{add}(y);\TEMPemit\,\tty{add}, \\
    \lambda y\ldotp\mathit{remove}(y);\TEMPemit\,\tty{remove}, \ 
    \lambda \_\ldotp\TEMPletin{r = \mathit{iter}()}{\TEMPemit \pair{\tty{iterator}}{r};r}, \\
    \lambda y\ldotp\TEMPletin{r = \mathit{next}(y)}{\TEMPemit \pair{\tty{next}}{y};r}
  \end{array} \right)
\end{multline*}
The traces ignore the arguments to \tty{add} and \tty{remove} and the return
values of \tty{size} and \tty{next}, as they are irrelevant for the protocol. 

With this instrumentation we can now express the informal protocol as a language
of permissible interaction traces between the client and library.
We let the trace alphabet be the countable set:
\[
\Sigma = \{ \tty{size},\tty{add},\tty{remove}\}\cup\{\pair{\tty{next}}{\ell},\pair{\tty{iterator}}{\ell}\mid \ell\in\mathit{Loc}\}
\]
The language $\Lcoll$ of safe behaviours contains all strings $t\in\Sigma^*$ such that, for all $1\leq i\leq|t|$:
\begin{quote}
if $t[i]=\pair{\tty{next}}{\ell}$ then there is $j<i$ such that $t[j]=\pair{\tty{iterator}}{\ell}$ and, for all $j<k<i$, $t[k]\notin\{\,\tty{add},\tty{remove}\,\}$. 
\end{quote}
That is, every call to the $\tty{next}$ method of an iterator $\ell$ must be preceded by a call 
to $\tty{iterator}()$ which returns $\ell$. In addition, there should be no modification of the collection 
between those two events.\!

To enforce this protocol we use a cut-down version of the iterator specification in \cite{Krishnaswami:2009:dp-sl}
that does not track the contents of the underlying collections.\vspace{-.75mm}
\begin{align*}
  \Phicoll(P_{0}, &(size, add, remove, iterator, next)) \eqdef \\
&\;
  \exists \collp : \ty{Val} \rightarrow \ty{Prop} \ldotp \exists \iterp : \ty{Val} \times \ty{Val} \rightarrow \ty{Prop}.    \valid(P_{0} \implies \exists c : \ty{Val}.~\collp(c))~\land \\
  &\hspace{1.2em}
  \begin{aligned}[t]
  &
    \forall c : \ty{Val} \ldotp \triple{\collp(c)}{size()}{\collp(c)}~\land \\
  &
    \forall c, x : \ty{Val} \ldotp \triple{\collp(c)}{add(x)}{\exists c' : \ty{Val} \ldotp \collp(c')}~\land \\
  &
    \forall c : \ty{Val} \ldotp \triple{\collp(c)}{remove(x)}{\exists c' : \ty{Val} \ldotp \collp(c')}~\land \\
  &
    \forall c : \ty{Val} \ldotp \triple{\collp(c)}{iterator()}{r.~\collp(c) * \iterp(r, c)}~\land \\
  &
    \forall c, x : \ty{Val} \ldotp \triple{\collp(c) * \iterp(x, c)}{next(x)}{\collp(c) * \iterp(x, c)}
  \end{aligned}
\end{align*}
The specification introduces two types of resources, $\collp$ and $\iterp$, to formally capture the
protocol. The $\collp(c)$ resource is indexed by an abstract version number $c$ while the $\iterp(p, c)$
resource expresses that $p$ is an iterator and the version number of the underlying collection was $c$
when the iterator was created. The operations that destructively update the collection consume a 
$\collp(c)$ resource and produce a $\collp(c')$ resource for an existentially quantified version
number $c'$. As a result, after
a destructive update, we should no longer be able to satisfy the precondition of \tty{next}, as it requires 
ownership of a $\collp$ resource and an $\iterp$ resource with a common version number $c$.

To establish this intuition formally we prove that, for an arbitrary library implementation $M$ that 
satisfies the collections specification, the wrapped library implementation $\wrapcoll(M)$ also
satisfies the collections specification \emph{and} the traces generated by the wrapped implementation
are in the language $\Lcoll$. 
\begin{lemma}
$P_{0}, \mathit{ops} \mid \Phicoll(P_{0}, \mathit{ops}) \vdel
\Phicoll(P_{0} * \trace(\varepsilon) * \inv(\Lcoll), \wrapcoll(\mathit{ops}))$.
\end{lemma}
\begin{proof}[sketch]
To prove that the wrapped implementation satisfies the collections specification and produces
traces in $\Lcoll$, we first need to define new wrapped versions of the $\collp$ and $\iterp$
resources that relate the abstract version number to the trace state. 
The idea is that the collection parameter of the wrapped $\collp$ resource will consist of a pair $(c,n)$, where the $c$ component is the parameter of the underlying $\collp$ resource and $n$ is the index in the trace of the last \tty{add} or \tty{remove} event. 
We want the wrapped $\collp((c,n))$ resource to assert that there are no \tty{add} or 
\tty{remove} events in the current trace after the $n$-th element. Likewise, the
wrapped $\iterp(r, (c,n))$ resource should assert that there \emph{is} an \tty{iterator} event for 
iterator $r$ in the current trace after the $n$-th element. Hence, if we own
both $\collp((c,n))$ and $\iterp(r, (c,n))$ then we know that no 
\tty{add} or \tty{remove} events were emitted since the iterator $r$ was created.

Let $\collp$ and $\iterp$ denote the non-wrapped representation
predicates that exist by the $\Phicoll(P_{0}, ops)$ assumption and
define the wrapped representation predicates $\collpw$ and $\iterpw$
as follows:
\begin{align*}
\collpw(x) &\eqdef 
  \begin{aligned}[t]
  &
    \exists y : \ty{Val} \ldotp \exists n : \ty{Nat} \ldotp x = (y, n) * \collp(y)*{} \\[-.75mm]
  &\quad
    \exists t \in \Lcoll \ldotp \trace(t) * \inv(\Lcoll) * \texttt{add}, \texttt{remove} \not\in t[n+1..] * n \leq |t|
  \end{aligned} \displaybreak[1]\\[-.5mm]
\iterpw(r, x) &\eqdef
  \begin{aligned}[t]
  &
    \exists y : \ty{Val} \ldotp \exists n : \ty{Nat} \ldotp x = (y, n) * \iterp(y)~*\\[-.75mm]
  &\quad
    \exists t \in \Lcoll \ldotp \hist(t) * \langle \texttt{iterator}, r \rangle \in t[n+1..]
  \end{aligned}
\end{align*}
We use $t[n..]$ as notation for the subtrace of $t$ starting from the
$n$-th element.

It thus remains to show that the wrapped library satisfies the collections specification. First,
we need to prove that we obtain a wrapped collection resource from the initial resources:
$P_{0} * \inv(\Lcoll) * \trace(\varepsilon) \implies \exists c : \ty{Val} \ldotp \collpw(c)$.
This follows easily from the $P_{0} \implies \exists c' : \ty{Val} \ldotp \collp(c')$ assumption
by taking the second component of $c$ to be 0.

Next, we have to show that each of the wrapped operations satisfies the corresponding Hoare
specification. The \tty{size} method is particularly simple, as any trace $t \in \Lcoll$ can
trivially be extended with a size event $t \cdot \texttt{size} \in \Lcoll$. We will thus skip
the \tty{size} method.
The \texttt{add} method is more interesting, as we have to update the index into the trace
for the last \texttt{add} event. Below is a proof outline for the wrapped \tty{add} method
applied to argument \tty{z}.
\begin{source}
[Context $c, \tty{z} : \ty{Val}$]
$\{ \collpw(c) \}$
$\{ \exists x, n, t \ldotp c = (x, n) * \collp(x) * \trace(t) * \inv(\Lcoll) * t\in\Lcoll$
  $*~\texttt{add}, \texttt{remove} \not\in t[n+1..] * n \leq |t| \}$
add(z);
$\{ \exists x', n, t \ldotp \collp(x') * \trace(t) * \inv(\Lcoll) * t\in\Lcoll
  *\texttt{add}, \texttt{remove} \not\in t[n+1..] * n \leq |t| \}$
emit add;
$\{ \exists x', n, t \ldotp \collp(x') * \trace(t \cdot \tty{add}) * \inv(\Lcoll) * t\in\Lcoll$
  $*~\texttt{add}, \texttt{remove} \not\in t[n+1..] * n \leq |t| \}$
$\{ \exists c' : \ty{Val} \ldotp \collpw(c') \}$
\end{source}
This leaves us with two proof obligations: firstly, we have to show that we are allowed to
emit the \tty{add} event (i.e., that $t \cdot \tty{add} \in \Lcoll$); and secondly for the
last step we have to show that:\vspace{-1mm}
\[
  \forall x', n, t \ldotp \left( \begin{array}{@{}l@{}}
    \collp(x') * \trace(t \cdot \tty{add}) * \inv(\Lcoll) * t\in\Lcoll \\ {} * n \leq |t|
    * \texttt{add}, \texttt{remove} \not\in t[n+1..] \end{array} \right)
\implies \exists c' : \ty{Val} \ldotp \collpw(c')
\]
This follows easily by taking $c'$ to be $(x', |t \cdot \tty{add}|)$, as $t \cdot \tty{add}$ contains
no \tty{add} or \tty{remove} events after the $|t \cdot \tty{add}|$-th element. 

The proof for \tty{remove} follows the same structure as for \tty{add}. 
For the \tty{iterator} method, we emit an \tty{iterator} event and create a new
$\hist$ resource to record the trace at the time of the creation of the iterator.
Below we give a proof outline for the \tty{iterator} method:\vspace{-1mm}
\begin{source}
[Context $c : \ty{Val}$]
$\{ \collpw(c) \}$
$\{ \exists x, n, t \ldotp c = (x, n) * \collp(x) * \trace(t) * \inv(\Lcoll) * t\in\Lcoll$
  $*~\texttt{add}, \texttt{remove} \not\in t[n+1..] * n \leq |t| \}$
let r = iterator() in
$\{ \exists x, n, t \ldotp c = (x, n) * \collp(x) * \iterp(\tty{r}, x) * \trace(t)$
  $*~\inv(\Lcoll) * t\in\Lcoll * \texttt{add}, \texttt{remove} \not\in t[n+1..] * n \leq |t| \}$
emit<<iterator, r>>;
$\{ \exists x, n, t \ldotp c = (x, n) * \collp(x) * \iterp(\tty{r}, x) * \trace(t \cdot \langle \tty{iterator}, \tty{r} \rangle)$
  $*~\inv(\Lcoll) * t\in\Lcoll * \texttt{add}, \texttt{remove} \not\in t[n+1..] * n \leq |t| \}$
$\{ \collpw(c) * \iterpw({\tty r}, c) \}$
r
$\{ r \ldotp \collpw(c) * \iterpw(r, c) \}$
\end{source}\vspace{-.5mm}
As before, we are left with two proof obligations: $t \cdot \langle \tty{iterator}, \tty{r} \rangle \in \Lcoll$ and:\vspace{-1mm}
\begin{align*}
\forall x, n, t \ldotp (
  \begin{aligned}[t]
  &
    c = (x, n) * \collp(x) * \iterp(\tty{r}, x) * \trace(t \cdot \langle \tty{iterator}, \tty{r} \rangle) *\inv(\Lcoll)  \\
  &{}
     * t\in\Lcoll * \texttt{add}, \texttt{remove} \not\in t[n+1..] * n \leq |t|) 
  \implies \collpw(c) * \iterpw(\tty{r}, c)
  \end{aligned}
\end{align*}
To discharge this last proof obligation, we use the \textsc{PAllocHist} rule to introduce a 
history resource $\hist(t \cdot \langle \tty{iterator}, \tty{r} \rangle)$ and since $n \leq |t|$
it follows that $\langle \tty{iterator}, \tty{r} \rangle \in (t \cdot \langle \tty{iterator, \tty{r} \rangle})[n+1..]$,
as required by $\iterpw(\tty{r}, c)$. 

We are left with \tty{next}, which is the most interesting case as it
requires us to prove the iterator we are trying to use is still valid. We give a
proof outline for the \tty{next} method applied to an argument $\tty{x}$:\vspace{-1mm}
\begin{source}
[Context $c,\tty{x} : \ty{Val}$]
$\{ \collpw(c) * \iterpw(\tty{x}, c) \}$
$\{ \exists y, n, t, t' \ldotp c = (y, n) * \collp(y) * \iterp(\tty{x}, c) * \trace(t) * \hist(t') * \inv(\Lcoll) $
  ${} * t,t'\in\Lcoll*\texttt{add}, \texttt{remove} \not\in t[n+1..] * n \leq |t| * \langle \texttt{iterator}, \texttt{x} \rangle \in t'[n+1..] \}$
let r = next(x);
$\{ \exists y, n, t, t' \ldotp c = (y, n) * \collp(y) * \iterp(\tty{x}, c) * \trace(t) * \hist(t') * \inv(\Lcoll) $
  ${} * t,t'\in\Lcoll*\texttt{add}, \texttt{remove} \not\in t[n+1..] * n \leq |t| * \langle \texttt{iterator}, \texttt{x} \rangle \in t'[n+1..] \}$
emit<<next, x>>;
$\{ \exists y, n, t, t' \ldotp c = (y, n) * \collp(y) * \iterp(\tty{x}, c) * \trace(t \cdot \langle \tty{next}, \tty{x} \rangle) * \hist(t') * \inv(\Lcoll) $
  ${} * t,t'\in\Lcoll*\texttt{add}, \texttt{remove} \not\in t[n+1..] * n\leq|t|*\langle \texttt{iterator}, \texttt{x} \rangle \in t'[n+1..] \}$
$\{ \collpw(c) * \iterpw(\tty{x}, c) \}$
r
$\{ \collpw(c) * \iterpw(\tty{x}, c) \}$
\end{source}\vspace{-.45mm}
To verify the emit expression, we further have to prove that $t \cdot \langle \tty{next}, \tty{x} \rangle \in \Lcoll$,
that is, that the iterator \tty{x} is still valid. This relies on the following key property of the $\Lcoll$ language:
\[
  t \in \Lcoll \land \texttt{add}, \texttt{remove} \not\in t[n+1..] 
  \land \langle \texttt{iterator}, \texttt{x} \rangle \in t[n+1..]
 {\implies} t \cdot \langle \texttt{next}, \texttt{x} \rangle \in \Lcoll
\]
To apply this property we use \textsc{PUseHist} to conclude from $\trace(t) * \hist(t')$ that
$t' \leqprefix t$ and thus that $\langle \texttt{iterator}, \texttt{x} \rangle \in t'[n+1..] \implies
\langle \texttt{iterator}, \texttt{x} \rangle \in t[n+1..]$. 
\end{proof}

\subsection{Well-bracketing Protocols}\label{sec:wb-example}

Libraries that allow clients to acquire, access and release resources often impose a
well-bracketing protocol whereby clients are required to acquire resources before 
accessing and releasing them. The file library in \S\ref{ex:file-lib} was a
particularly simple example of such a protocol. In this section we consider a more
advanced and realistic variant thereof for a library
with a higher-order function that takes care of acquiring and releasing the underlying
resource for clients. 

Consider a library with a higher-order method \tty{withRes} for acquiring, accessing
and subsequently releasing some resource (e.g.\ a file) and an operation \tty{op} for
accessing the resource. The \tty{withRes} operation takes as argument a function $f$
provided by the client for accessing the resource and takes care of acquiring the resource
before $f$ is called and subsequently releasing it again. For such a library, we wish
to ensure that 1) clients only access the resource after they have acquired it, and
2) clients do not try to acquire resources they already hold. 

To express this property formally, we first define a wrapping function
that instruments the library to emit events about the interaction between client
and library. Formally, we take a library implementation to be a tuple consisting
of a \tty{withRes} function and an operation \tty{op}. Below we define a wrapping
function for such a library that emits call and return events for all calls where
control passes between client and library.
\begin{align*}
\wrapbrack \eqdef 
  \lambda (\mathit{withRes}, \mathit{op}) \ldotp (
    &\; \lambda f \ldotp
      \begin{aligned}[t]
      &
        \emit \langle \tty{call}, \tty{withRes}, f \rangle; \\
      &
        \mathit{withRes}(\lambda x.~\emit \langle \tty{call}, f \rangle; f(x);  
       \emit \langle \tty{ret}, f \rangle);\\
      &  \emit \langle \tty{ret}, \tty{withRes}, f \rangle, 
      \end{aligned} \\
    &\; \lambda \tty{x} \ldotp \emit \langle \tty{call}, \tty{op} \rangle; \mathit{op}(\tty{x}); \emit \langle \tty{ret}, \tty{op} \rangle\ )
\end{align*}
We can now state the desired property as a well-bracketing property
of the traces generated by the instrumented library. Let us define $\Lbrack\subseteq\VAL^*$ 
as the prefix closure of the language of all strings $s$ that are of the form:
\[
\pair{\tty{call}}{\tty{withRes},f}\cdot\pair{\tty{call}}{f} \cdot s_{\tt op} \cdot
\pair{\tty{ret}}{f} \cdot \pair{\tty{ret}}{\tty{withRes},f} \cdot s'
\] 
for some $f\in\FUN$, $s_{\tt op}\in(\pair{\tty{call}}{\tty{op}} \cdot \pair{\tty{ret}}{\tty{op}})^*$ and $s'\in \Lbrack$.
That is, the strings in $\Lbrack$ are well-bracketed sequences of events formed of subsequences adhering to the pattern $\pair{\tty{call}}{\tty{withRes},f} \cdot \pair{\tty{call}}{f} \cdot
s_{\tt op} \cdot \pair{\tty{ret}}{f} \cdot \pair{\tty{ret}}{\tty{withRes},f}$, and subsequences thereof, where $s_{\tt op}$ a sequence of consecutive calls and returns of $\tt op$.

This trace property is enforced by the following separation logic
specification. 
\begin{align*}
&\Phibrack(P_{0}, (\tty{withRes}, \tty{op})) \eqdef \\
&\;
  \begin{aligned}[t]
  &
    \exists \lockedp : \ty{Prop} \ldotp \exists \unlockedp : \ty{Val} {\rightarrow} \ty{Prop} \ldotp
    \valid(P_{0} {\implies} \lockedp)~\land \\
  &\;\;
    \forall P, Q : \ty{Prop} \ldotp \forall f : \ty{Val} \ldotp \\
  &\quad
      \{ \lockedp * P * S(P, Q, \unlockedp, f) \} 
      \tty{withRes}(f) 
      \{ \lockedp * Q \}  \land{}\\
  &\quad
    \forall x, y : \ty{Val} \ldotp \{ \unlockedp(y) \}~\tty{op}(x)~\{ \unlockedp(y) \}
  \end{aligned} \\
&S(P, Q, \unlockedp, f) \eqdef \forall y,x : \ty{Val} \ldotp \spec(\{ \unlockedp(y) * P \}~f(x)~\{ \unlockedp(y) * Q \})
\end{align*}
This uses two abstract resources, $\unlockedp$ and $\lockedp$ to
capture the well-bracketing aspect of the protocol. In particular, calling
\tty{withRes} requires the client to relinquish ownership of the \lockedp{} resource.
Since the function provided by the client is only given ownership of the
abstract \unlockedp{} resource, it cannot itself call \tty{withRes}. 
Furthermore, to call \tty{op} requires ownership of the \unlockedp{} resource,
thus ensuring that only the callback provided by the client to \tty{withRes}
can call \tty{op}. This specification ensures that \tty{withRes} is forced
to call the function provided by the client \emph{exactly} once, as it is
required to transform the abstract resource $P$ into $Q$ and the only way
it can achieve this is by calling the function provided by the client. 

To prove that the specification enforces the trace property we proceed as usual,
by proving that for any implementation $M$ that satisfies the specification, the
wrapped implementation $\wrapbrack(M)$ also satisfies the specification and the
traces generated by the wrapped implementation are in $\Lbrack$. 
\begin{lemma}
$P_{0}, ops \mid \Phibrack(P_{0}, ops) \vdel
\Phibrack(P_{0} * \trace(\varepsilon) * \inv(\Lbrack), \wrapbrack(ops))$.
\end{lemma}
\begin{proof}[Proof sketch]
To prove that the wrapped version satisfies the specification and produces traces in
$\Lbrack$, we first need to define wrapped versions of the abstract representation
predicates. The idea is to let the wrapped $\lockedp$ resource, $\lockedpw$, express 
that the current trace $t$ is in $\Lbrack$ and the $\langle \tty{call}, \tty{withRes}, f \rangle$
and $\langle \tty{ret}, \tty{withRes}, f \rangle$ events in $t$ are well-balanced and well-bracketed. 
For the wrapped $\unlockedp$ resource, $\unlockedpw(x)$, the idea is to use the
argument $x$ to track the name $f$ of the last unbalanced $\langle \tty{call}, \tty{withRes}, f \rangle$
event in $t$.

To simplify the definitions and subsequent proofs, we first introduce a number of auxiliary
resources, $T_0, T_1, T_2$ and $T_3$. $T_0$ expresses that the current trace is well-balanced. 
We set ${\cal O}=(\langle \tty{call}, \tty{op} \rangle \cdot \langle \tty{ret}, \tty{op} \rangle)^*$.
$T_1(f)$ expresses that the current trace $t$ has the form
$s \cdot \langle \tty{call}, \tty{withRes}, f \rangle$ where $s$ is well-balanced. $T_2(f)$
expresses that the current trace $t$ has the form 
$s \cdot \langle \tty{call}, \tty{withRes}, f \rangle \cdot \langle \tty{call}, f \rangle \cdot s'$
where $s$ is well-balanced and $s' \in \cal O$.
Finally, $T_3(f)$ expresses that the current trace $t$ has the form 
$s \cdot \langle \tty{call}, \tty{withRes}, f \rangle \cdot \langle \tty{call}, f \rangle \cdot s' \cdot \langle \tty{ret}, f \rangle$
where $s$ is well-balanced and $s' \in \cal O$.\vspace{-3mm}

\begin{align*}
T_0 &= 
  \exists t \in \Lbrack \ldotp \trace(t) * \inv(\Lbrack) 
  * (|t| > 0 \implies \exists f \ldotp t[|t|] = \langle \tty{ret}, \tty{withRes}, f \rangle) \\
T_1(f) &=
  \exists t \in \Lbrack \ldotp \trace(t) * \inv(\Lbrack)
  * t[|t|] = \langle \tty{call}, \tty{withRes}, f \rangle \\
T_2(f) &=
  \begin{aligned}[t]
  &
    \exists t \in \Lbrack \ldotp \trace(t) * \inv(\Lbrack) * \exists n < |t| \ldotp \
    t[n] = \langle \tty{call}, \tty{withRes}, f \rangle \\
&\qquad\qquad\quad\qquad\qquad\quad\qquad\qquad
\land t[n+1] = \langle \tty{call}, f \rangle\land
    t[n+2..] \in \cal O
  \end{aligned} 
\end{align*}
\begin{align*}
T_3(f) &=
  \begin{aligned}[t]
  &
    \exists t \in \Lbrack \ldotp \trace(t) * \inv(\Lbrack) * \exists n < |t| \ldotp \
    t[n] = \langle \tty{call}, \tty{withRes}, f \rangle \\
&
\land t[n+1] = \langle \tty{call}, f \rangle\land 
    t[n+2..(|t|-1)] \in {\cal O}\land 
    t[|t|] = \langle \tty{ret}, f \rangle
  \end{aligned}
\end{align*}
With these resources, we can now define $\unlockedpw$ and $\lockedpw$:
\begin{align*}
\unlockedpw(x) &\eqdef \exists y, z : \ty{Val} \ldotp x = (y, z) * \unlockedp(y) * T_2(z) \\
\lockedpw &\eqdef \lockedp * T_0
\end{align*}
It follows easily that 
$P_{0} * \trace(\varepsilon) * \inv(\Lbrack) \implies \lockedpw$
from $\trace(\varepsilon) * \inv(\Lbrack) \implies T_0$ and the assumption $P_{0} \implies \lockedp$.3

It remains to show that the two instrumented operations satisfy their
specifications. We begin by showing that the instrumented \tty{withRes} operation
satisfies its specification:
\begin{align*}
&
  \forall P, Q : \ty{Prop} \ldotp \forall f : \ty{Val} \ldotp \\
&\{ \lockedpw * P * S(P, Q, \unlockedpw, f) \} 
    \pi_1(\wrapbrack(\tty{withRes}, \tty{op}))(f) \;\
    \{ \lockedpw * Q \}
\end{align*}
assuming \tty{withRes} satisfies its specification:
\[
  \forall P, Q : \ty{Prop} \ldotp \forall f : \ty{Val} \ldotp \\
    \{ \lockedp * P * S(P, Q, \unlockedp, f) \}~
    \tty{withRes}(f)~
    \{ \lockedp * Q \}
\]
We give a proof outline for the instrumented \tty{withRes} operation:
\begin{source}
[Context $P, Q : \ty{Prop}, f : \ty{Val}$]
$\{ \lockedpw * P * S(P, Q, \unlockedpw, f) \}$
$\{ \lockedp * T_0 * P * S(P, Q, \unlockedpw, f) \}$
  emit<<call, withRes, f>>;
$\{ \lockedp * T_1(f) * P * S(P, Q, \unlockedpw, f) \}$
  let g = lambda x. emit<<call, f>>; f(x); emit<<ret, f>>  in
$\{ \lockedp * T_1(f) * P * S(P * T_1(f), Q * T_3(f), \unlockedp, g) \}$
  withRes(g);
$\{ \lockedp * T_3(f) * Q \}$
  emit<<ret, withRes, f>>
$\{ \lockedp * T_0 * Q \}$
$\{ \lockedpw * Q \}$
\end{source}
The interesting step is showing that from the assumed specification of $f$ we can
derive the desired specification for the instrumented version of $f$:
\begin{align*}
  \forall P, Q :  & \ty{Prop}\ldotp \forall f : \ty{Val} \ldotp 
    S(P, Q, \unlockedpw, f)\\ 
&\hspace{-1em} \implies 
  \begin{aligned}[t]
    S(&P * T_1(f), Q * T_3(f), \unlockedp, 
    \lambda x.~\emit \langle \tty{call}, f \rangle; f(x); \emit \langle \tty{ret}, f \rangle) 
  \end{aligned}\\[-5mm]
\end{align*}
This follows from the following proof outline
\begin{source}
[Context $P, Q : \ty{Prop}$ and $f, x,y : \ty{Val}$]
$\{ \unlockedp(y) * P * T_1(f) * S(P, Q, \unlockedpw, f) \}$
  emit<<call, f>>;
$\{ \unlockedp(y) * P * T_2(f) * S(P, Q, \unlockedpw, f) \}$
$\{ \unlockedpw((y, f)) * P * S(P, Q, \unlockedpw, f) \}$
  f(x);
$\{ \unlockedpw((y, f)) * Q \}$
$\{ \unlockedp(y) * Q * T_2(f) \}$
  emit<<ret, f>>
$\{ \unlockedp(y) * Q * T_3(f) \}$
\end{source}
Lastly, we need to show that the instrumented \tty{op} function satisfies its specification.
Below we give a proof outline for the instrumented \tty{op} function applied to an argument
$x$:
\begin{source}
[Context $x, y : \ty{Val}$]
$\{ \unlockedpw(y) \}$
$\{ \exists a, f \ldotp y = (a, f) * \unlockedp(a) * T_2(f) \}$
  emit<<call, op>>;
$\{ \exists a, f \ldotp y = (a, f) * \unlockedp(a) * \exists t \in \Lbrack \ldotp \trace(t) *~\inv(\Lbrack) * p(t,f) \}$
  op(x);
$\{ \exists a, f \ldotp y = (a, f) * \unlockedp(a) * \exists t \in \Lbrack \ldotp \trace(t) *~\inv(\Lbrack) * p(t,f) \}$
  emit<<ret, op>>
$\{ \exists a, f \ldotp y = (a, f) * \unlockedp(a) * T_2(f) \}$
$\{ \unlockedpw(y) \}$
\end{source}
where $p(t,f)=
\exists n\,{<}\,|t| \ldotp\, 
    t[n] = \langle \tty{call}, \tty{withRes}, f \rangle \land t[n+1] = \langle \tty{call}, f \rangle\land{}$ 
$    t[n+2..] \in (\langle \tty{call}, \tty{op} \rangle \cdot \langle \tty{ret}, \tty{op} \rangle)^* \cdot \langle \tty{call}, \tty{op} \rangle$. 
\qed
\end{proof}

\subsection{Traversable stack example}\label{stackmap-example}

\newcommand{\stacktre}{\ty{stk_{tr}}}
\newcommand{\stackp}{\ty{stack}}
\newcommand{\stackpw}{\ty{stack_w}}
\newcommand{\wrapstack}{\wrap_{\textsf{stack}}}

To further demonstrate that our approach can express and enforce strong trace
properties, recall the stack
example from \S\ref{sec:mot-ex}. We have a stack with a \texttt{push} and a
\texttt{pop} method, and a \texttt{foreach} method that takes a function argument and applies
the given function to every element of the stack, in order, starting from the top-most
element. 
Here the protocol on the interaction between
client and library imposes restrictions on both the client and the library. In
particular, we wish to ensure that the function provided by the client 
cannot call back into the stack-library and potentially modify the underlying
stack during the iteration of the stack. We also wish to ensure that the
library calls the function provided by the client with every element currently
on the stack and in the right order. 

To express this protocol, we first define a suitable library wrapper that
tracks all calls to \texttt{push} and \texttt{pop} and all calls to the function argument
provided by the client when calling \texttt{foreach}. \vspace{-1mm}
\begin{align*}
\wrap&_{\textsf{stack}}(\mathit{push}, \mathit{pop}, \mathit{foreach})\eqdef\\
  (\lambda a.~&
    \emit \langle \tty{call}, \texttt{push}, a \rangle; \mathit{push}(a); \emit \langle \tty{ret}, \texttt{push} \rangle, \\[-3pt] 
  \lambda \_.~&
    \emit \langle \tty{call}, \texttt{pop} \rangle; \TEMPletin{x = \mathit{pop}()}{\emit \langle \tty{ret}, \texttt{pop}, x \rangle; x}, \\[-3pt]
  \lambda f.~&
    \emit \langle \tty{call}, \texttt{foreach}, f \rangle;
    \mathit{foreach}\bigl(\lambda a.~\emit \langle \tty{call}, f, a \rangle; f(a);  \\[-3pt]
  &
    \emit \langle \tty{ret}, f \rangle\bigr);  \emit \langle \tty{ret}, \texttt{foreach} \rangle
  )\\[-5mm]
\end{align*}
Let $\Sigma_{\sf st}$ be the stack alphabet.
We can formalize the intended protocol as the language ${\cal L}_{\sf stack}$ defined as the prefix closure of the language of all traces $t\in\mathit{Val}^*$ such that $\stacktre(t,\varepsilon)$ holds, where, given $\alpha\in\Sigma_{\sf st}^*$,  we define $\stacktre(t,\alpha)$ by:
\begin{align*}
&\stacktre(t,\alpha)\, \eqdef\ (t=\alpha=\varepsilon)\lor{}
\\
&\; ( t= t'\!\cdot\! \langle\tty{call}, \texttt{push}, a \rangle \cdot\! \langle \tty{ret}, \texttt{push} \rangle
\land \alpha\,{=}\,a\,{::}\,\alpha'\!\land\stacktre(t'\!,\alpha'))\lor\\[-3pt]
&\;( t= t'\cdot
\langle \tty{call}, \texttt{pop} \rangle \cdot \langle \tty{ret}, \texttt{pop}, () \rangle\land \alpha=\varepsilon\land\stacktre(t',\varepsilon))\lor{}
\\
&\;( t= t'\cdot
\langle \tty{call}, \texttt{pop} \rangle \cdot \langle \tty{ret}, \texttt{pop}, a \rangle\land\stacktre(t',a::\alpha))\lor{}
\\
&\;( t=t' \cdot \langle \tty{call}, \texttt{foreach}, f \rangle \cdot t'' \cdot \langle \tty{ret}, \texttt{foreach} \rangle\land\stacktre(t',\alpha)
\land \textsf{trav}(t'',\alpha,f))
\\
&\textsf{trav}(t,\alpha,f)\, \eqdef\ (t=\alpha=\varepsilon)\lor{}
\\
&\quad(t=\langle \tty{call}, f,a \rangle \cdot \langle \tty{ret}, f \rangle \cdot t'\land\alpha=a::\alpha'\land\textsf{trav}(t',\alpha',f))
\end{align*}
%
%
%
A higher-order separation logic specification for
such a stack data structure is the following.\vspace{-1mm}
\begin{align*}
&
  \Phi(P_{init}, (\mathit{push}, \mathit{pop}, \mathit{foreach})) \eqdef \\
&\;
  \begin{aligned}[t]
  &
    \exists \stackp : \ty{Val~seq} \rightarrow \ty{Prop} \ldotp \valid(P_{init} \implies \stackp(\varepsilon))~\land \\
  &\;
    \forall \alpha, a \ldotp \{ \stackp(\alpha) \land a \neq () \}~\mathit{push}(a)~\{ \stackp(a::\alpha) \}~\land \\
  &\;
    \forall \alpha \ldotp
    \{ \stackp(\alpha) \}~\mathit{pop}()~\{ r.~
      \begin{aligned}[t]
      &
        (r = () \land \stackp(\alpha) \land \alpha = \varepsilon) \\
      &\hspace{-2em}
        \lor~(\exists \alpha' \ldotp \alpha = r :: \alpha' \land \stackp(\alpha'))
      \} ~\land
      \end{aligned} \\
    &\;
      \forall \alpha, f, I \ldotp
      \begin{aligned}[t]
      &
        \{ \stackp(\alpha) * I(\varepsilon) * \forall \beta, a \ldotp \spec(\{ I(\beta) \}~f(a)~\{ I(a::\beta) \}) \} \\
      &\quad
        \mathit{foreach}(f) \{ \stackp(\alpha) * I(\mathit{rev}(\alpha)) \}
      \end{aligned}
  \end{aligned}
\end{align*}
 It asserts existence of an
abstract stack representation predicate $\stackp(\alpha)$ that tracks the 
exact sequence of elements currently on the stack using the mathematical
sequence $\alpha$. The specification for \texttt{push} and \texttt{pop} is straightforward:
pushing and popping elements pushes or pops elements from this mathematical 
sequence, with a few special cases for pushing $()$ or popping from an empty
stack. The specification for \texttt{foreach} is more interesting. It is parametrised
by a predicate $I$, to be chosen by the client. This predicate is indexed
by a sequence $\alpha$ and $I(\alpha)$ is intended to capture the client's
state after the function provided by the client has been called on each 
element of $\alpha$, in reverse order. This accounts for the $I(\mathit{rev}(\alpha))$
in the post-condition, where $\mathit{rev}$ is the reverse operator on sequences. 

\begin{lemma}
$P_{0}, \mathit{ops} \mid \Phi(P_{0}, \mathit{ops}) \vdel
\Phi(P_{0} * \inv(\mathcal{L}_{\sf stack}) * \trace(\varepsilon), \wrapstack(\mathit{ops}))$.
\end{lemma}\vspace{-3mm}
\begin{proof}[Proof sketch]
We proceed by defining a wrapped version of the \stackp{} predicate
that asserts that the sequence of elements $\alpha$ matches the expected
contents of the stack as per the current trace $t$.
\begin{align*}
\stackpw(\alpha) &\eqdef
  \stackp(\alpha) * \exists t\ldotp \ty \stacktre(t,\alpha)* \trace(t) * \inv(\mathcal{L}_{\stackp})
\end{align*}
Clearly we have that
$P_{init} * \inv(\mathcal{L}_{\stackp}) * \trace(\varepsilon) \implies \stackpw(\varepsilon)$
as $\stacktre(\varepsilon,\varepsilon)$ and $P_{init} \implies \stackp(\varepsilon)$.

It remains to prove the wrapped library methods satisfy the specification
instantiated with the wrapped \stackp{} predicate. The proofs for \texttt{push} 
and \texttt{pop} are straightforward and have been omitted. For \texttt{foreach}
we are given a predicate $I$ from the client and need to prove the following
triple:
\begin{align*}
&
  \{ \stackpw(\alpha) * I(\varepsilon) * \forall \beta, a \ldotp \spec(\{ I(\beta) \}~f(a)~\{ I(a :: \beta) \}) \} \\
&\quad
  \pi_3(\wrapstack(\mathit{ops}))(f) \{ \stackpw(\alpha) * I(rev(\alpha)) \}
\end{align*}
In the call to the underlying \texttt{foreach} method, we can pick a suitably
wrapped version of the $I$ predicate, $I_{\sf w}(\beta)$. The idea is that it should
assert $I(\beta)$ and that we have emitted an opening \texttt{foreach} call
and called the function argument on all the elements of $\beta$ so far. 
\begin{align*}
I_{\sf w}(\beta) \eqdef
  I(\beta)  * \exists t_1,t_2\ldotp\,& \stacktre(t_1,\alpha)\land\mathsf{trav}(t_2,rev(\beta),f) 
    {}\land \trace(t_1 \cdot \langle \tty{call}, \texttt{foreach}, f \rangle \cdot t_2)
\end{align*}
We need to prove that the wrapped function argument updates the
wrapped $I$ predicate appropriately. This follows from:\vspace{-1mm}
\begin{source}
[Context $\alpha,\beta,f : \ty{Val}$]
$\{ I_{\sf w}(\beta) \}$
$\{ I(\beta) * \exists t_1,t_2\ldotp \stacktre(t_1,\alpha) * \mathsf{trav}(t_2,rev(\beta), f)
 * \trace(t_1 \cdot \langle \tty{call}, \texttt{foreach}, f \rangle \cdot t_2) \}$
emit <<call, $f$, $a$>>;
$\{ I(\beta) * \exists t_1,t_2\ldotp\stacktre(t_1,\alpha) * \mathsf{trav}(t_2,rev(\beta), f)$
$\qquad\qquad{} * \trace(t_1 \cdot \langle \tty{call}, \texttt{foreach}, f \rangle \cdot t_2 \cdot \langle \tty{call}, f, a \rangle) \}$
f(a);
$\{ I(a :: \beta) * \exists t_1,t_2\ldotp\stacktre(t_1,\alpha) * \mathsf{trav}(t_2,rev(\beta), f)$
$\qquad\qquad {}* \trace(t_1 \cdot \langle \tty{call}, \texttt{foreach}, f \rangle \cdot t_2 \cdot \langle \tty{call}, f, a \rangle) \}$
emit <<ret, $f$>>;
$\{ I(a :: \beta) * \exists t_1,t_2\ldotp\stacktre(t_1,\alpha) * \mathsf{trav}(t_2,rev(\beta), f)$
$ \qquad {} * \trace(t_1 \cdot \langle \tty{call}, \texttt{foreach}, f \rangle \cdot t_2 \cdot \langle \tty{call}, f, a \rangle \cdot \langle \tty{ret}, f \rangle) \}$
$\{ I(a :: \beta) * \exists t_1,t_2\ldotp\stacktre(t_1,\alpha) * \mathsf{trav}(t_2,rev(a::\beta), f) * \trace(t_1 \cdot \langle \tty{call}, \texttt{foreach}, f \rangle \cdot t_2) \}$
$\{ I_{\sf w}(a :: \beta) \}$
\end{source}\vspace{-1mm}
The second to last step follows from the following property:
$$\forall \alpha , a , t\ldotp 
\textsf{trav}(t,\alpha,f)\implies
\textsf{trav}(t \cdot \langle \tty{call}, f, a \rangle \cdot \langle\tty{ret}, f \rangle,\alpha \cdot a, f)$$
Now the rest of the proof of \texttt{foreach} is just an application of the
specification of the underlying \texttt{foreach} method and the \textsc{Emit}
rule for the emission of the \texttt{foreach} call and return events. 
\end{proof}

\section{Semantics}
\label{sec:semantics}

We give a denotational semantics for the logic introduced previously 
and establish soundness of the logic.
The semantics is based on an interpretation of
resources as members of a suitable resource monoid $\sf M$. In fact,
$\sf M$ is a partial commutative monoid
$(|{\sf M}|,\bullet,1)$ whereby the multiplication operator acts as
the semantic analogue of separating conjunction. We shall also use
the partial order relation yielded by monoid multiplication:
$m_1\leq m_2 \iff \exists m\in|{\sf M}|.~m_1\bullet m=m_2$
for all $m_1,m_2\in|{\sf M}|$. 
In the sequel we shall frequently abuse
notation and write $|\sf M|$ simply as $\sf M$. 
The semantics of resources is parametrised on \emph{worlds}, 
that is, partially ordered sets $\ty{W}$ specifying functional environments and trace invariants.

The monoid models heap resources and trace resources.
It is constructed as the product of two partial commutative monoids: \
$\MEL\, \eqdef\, \HEAP \times \mathit{Trace}$,\, where:
\begin{align*}
\HEAP\, &\eqdef\, (\HEAP, \uplus, \emptyset) &&&
\mathit{Trace}\, &\eqdef\, (\{\botHist,\topTrace\} \times \VAL^*, {\bullet}\,, (\botHist, \varepsilon))
\end{align*}
The monoid multiplication on $\HEAP$ is disjoint union ($\uplus$), which is only defined between heaps with disjoint domains.
The empty heap $\emptyset$ is the unit.

The monoid multiplication for the trace monoid is defined by:
\begin{align*}
  (\botHist, t_1) \bullet (\botHist, t_2) &\eqdef \begin{cases}
            (\botHist, t_1) & \text{if } t_2 \leqprefix t_1 \\
            (\botHist, t_2) & \text{if } t_1 \leqprefix t_2 \\
            \text{undefined} & \text{otherwise} \end{cases} \\
  (\topTrace, t_1) \bullet (\botHist, t_2) &\eqdef (\botHist, t_2) \bullet (\topTrace, t_1) \eqdef \begin{cases}
            (\topTrace, t_1) & \text{if } t_2 \leqprefix t_1 \\
            \text{undefined} & \text{otherwise} \end{cases} \\
  (\topTrace, t_1) \bullet (\topTrace, t_2) & \quad \text{undefined}
\end{align*}
where $\leqprefix$ is prefix ordering on finite sequences.
The idea is that $\trace(t)$ is modelled by $(\topTrace, t)$, while $\hist(t)$ is modelled by $(\botHist, t)$.
The monoid multiplication ensures that $(\topTrace, t)$ is a unique resource, which grants the right to extend the trace.
(It does not permit arbitrary changes to the trace: updates must preserve all frames, and hence all prefixes of the trace.)
The resource $(\botHist, t)$ is duplicable, and only ensures that $t$ is a prefix of the trace.
It is easy to see that $(\botHist, \varepsilon)$ (where $\varepsilon$ denotes the empty sequence) is the unit of this monoid.

Worlds model the information contained in assertions that does not behave like a resource.
Worlds include the function environment and a trace invariant. We thus define
  $\WEL \eqdef \FENV \times \mathcal{P}(\VAL^*)$
with
$  (\gamma_1, I_1) \leq (\gamma_2, I_2)$ iff
$\gamma_1 \subseteq \gamma_2$ and $I_1 = I_2\,.
$
The ordering on worlds describes how they may change with time.
This ordering allows new functions to be named, but enforces that the trace invariant does not change with time.

Assertions are interpreted as monotone functions from worlds to
upwards closed sets of resources: $\semel{\ty{Prop}} \eqdef\, \WEL \monarrow \mathcal{P}^{\uparrow}(\MEL)$
where  $\mathcal{P}^{\uparrow}({\MEL}) \eqdef\,\{\,p\subseteq \MEL\mid\forall m\in p.\,\forall m'\in \MEL.~m\leq m'\implies m'\in p\,\}$.
The rest of the types are interpreted 
as shown below, where the ordering on $\{ \bot, \top \}$
is $\bot < \top$. 
\begin{align*}
\semel{1} &\eqdef \{*\} &
\semel{\ty{Bool}} &\eqdef \{\ty{true},\ty{false}\}&
\semel{\ty{Loc}} &\eqdef \LOC &
\semel{\tau\to\sigma} &\eqdef \semel{\tau}\to\semel{\sigma} 
\\
\semel{\ty{Nat}} &\eqdef \mathbb{N} &
\semel{\ty{Val}} &\eqdef \VAL &
\semel{\ty{Exp}} &\eqdef \EXP &
\semel{\tau\times\sigma} &\eqdef \semel{\tau}\times\semel{\sigma} &
\\
&&
\semel{\ty{seq}\,\tau} &\eqdef (\semel{\tau})^* &&&
\semel{\ty{Spec}} &\eqdef \WEL \monarrow \{\bot,\top\} 
\end{align*}

\newcommand\qwe{\;}
\newcommand\qwee{\hspace{-.3em}}

\begin{figure*}[t]
\[
\begin{array}{r@{}l}
  \semel{\Gamma \vdash P \Rightarrow Q : \ty{Prop}}_\rho(w) & \eqdef 
\Setb{m}{\begin{array}{@{}l} \forall m'\geq m \ldotp \forall w' \geq w \ldotp m'\in \semel{\Gamma \vdash P : \ty{Prop}}_\rho(w') \\
\implies\, m' \in \semel{\Gamma \vdash Q : \ty{Prop}}_\rho(w')\end{array} } 
\\
  \semel{\Gamma \vdash P * Q : \ty{Prop}}_\rho(w) &
\multicolumn{1}{l}{\eqdef  
\Setb{m}{
\begin{array}{@{}l} 
\exists m_1, m_2 \ldotp
m_1 \in \semel{\Gamma \vdash P : \ty{Prop}}_\rho(w)
\\ \land m_2 \in \semel{\Gamma \vdash Q : \ty{Prop}}_\rho(w) 
\land m=m_1 \bullet m_2 
\end{array}} 
}
\\
  \semel{\Gamma \vdash M \mapsto N : \ty{Prop}}_\rho(w) &
\multicolumn{1}{l}{\eqdef  
    \{\,m\in\MEL\mid m(\semel{\Gamma \vdash M : \ty{Loc}}_\rho) = \semel{\Gamma \vdash N : \ty{Val}}_\rho \}
}
\\
  \semel{\Gamma \vdash \spec(M) : \ty{Prop}}_\rho(w) &
\multicolumn{1}{l}{\eqdef  
    \{\,m\in\MEL\mid\semel{\Gamma \vdash M : \ty{Spec}}_\rho(w) = \top\,\}
}
\\
  \semel{\Gamma \vdash \valid(P) : \ty{Spec}}_\rho(w) &
\multicolumn{1}{l}{\eqdef  
    (\,\semel{\Gamma \vdash P : \ty{Prop}}_\rho(w) = \MEL\,)
}
\\
  \semel{\Gamma\vdash \trace(t) : \ty{Prop}}_\rho(w) &
\multicolumn{1}{l}{\eqdef  
\Setb{ (h,\tau) }{ \tau \geq (\topTrace,\semel{\Gamma \vdash t : \ty{seq}~\ty{Val}}_\rho) }
} 
\\
  \semel{\Gamma\vdash \hist(t) : \ty{Prop}}_\rho(w) &
\multicolumn{1}{l}{\eqdef  
\Setb{ (h,\tau) }{ \tau \geq (\botHist,\semel{\Gamma \vdash t : \ty{seq}~\ty{Val}}_\rho) } 
}
\\
  \semel{\Gamma\vdash \inv(I) : \ty{Prop}}_\rho(w) &
\multicolumn{1}{l}{\eqdef  
\Setb{ m\! }{ \!\pi_2(w) = \Setb{t\!}{\!\semel{\Gamma \vdash I : \ty{seq}~\ty{Val} \to \ty{Bool}}_\rho (t) = \ty{true}} \!} 
}
\\
  \semel{\Gamma \vdash \triple{P}{e}{Q} : \ty{Spec}}_\rho(w) &
\multicolumn{1}{l}{\eqdef  \begin{array}{@{}l@{}}
\forall w' \geq w \ldotp \semel{\Gamma \vdash P : \ty{Prop}}_\rho(w') \\[-.5mm]
\hspace{4em} \subseteq \mathsf{wp}(\semel{\Gamma \vdash Q : \ty{Val} \to \ty{Prop}}_\rho)(\rho(e))(w') \end{array}
}\\
  \semel{\Gamma \vdash M=_\tau N : \ty{Spec}}_\rho(w) &
\multicolumn{1}{l}{\eqdef  
    (\,\semel{\Gamma \vdash M : \tau}_\rho = \semel{\Gamma \vdash N : \tau}_\rho\,)
}
\\[-1mm]
\end{array}
\]
\[
  \mathsf{wp}(Q) \eqdef \nu \mathit{wp}' \ldotp \lambda e, w \ldotp
    \Setb{m}{ \begin{array}{@{}l@{}}
        \forall r, t, s \ldotp~ t, s \vDash_w m \bullet r \implies \\[-2pt]
                \quad ( e,s \nrightarrow {} \implies e \in \VAL \land m \in Q(e)(w)) \\[-2pt]
                {} \land \forall a, e', s' \ldotp~ e, s \xrightarrow{a} e', s' 
                \implies \\[-2pt] \qquad \exists w' \geq w, m' \ldotp (t \cdot a), s' \vDash_{w'} m' \bullet r
                \land m' \in \mathit{wp}'(e')(w')
              \end{array}}
\]
\begin{align*}\\[-8.5mm]
t, (h, \gamma) \vDash_{(\gamma', I)} (h', (\topTrace, t')) &\eqdef (t = t' \land h = h' \land \gamma = \gamma' \land t \in I) \\[-2mm]
  t, (h, \gamma) \vDash_{(\gamma', I)} (h', (\botHist, t')) &\eqdef (t' \leqprefix t \land h = h' \land \gamma = \gamma' \land t \in I)
\end{align*}\vspace{-7.25mm}
\caption{Semantics of terms (selected cases).}\label{fig:semEL}\vspace{-3mm}
\end{figure*}

The semantics of a term $\Gamma \vd M : \tau$ is defined inductively as in Figure~\ref{fig:semEL}.
(We give selected cases; for full details see Appendix~\ref{sec:elsound}.) The semantics is
defined in terms of a variable environment
$\rho \in \semel{\Gamma}$ that maps variables of type $\tau$ to elements of $\semel{\tau}$:
$\semel{\Gamma} =
  \Setb{ \rho : \dom{\Gamma}\rightarrow \bigcup\nolimits_{\tau} \semel{\tau} }{ \forall (x : \tau) \in \Gamma.~\rho(x) \in \semel{\tau} }$

The definition of the weakest precondition $\mathsf{wp}(Q)$, used to define the semantics of Hoare triples,
enforces the invariant on traces generated by the considered terms.
It is defined as a greatest fixed-point, which establishes that the updates to the concrete state (the trace $t$, heap $h$, and function context $\gamma$) simulate updates to the abstract state (the resource $m$ and world $w$), with respect to the erasure relation ($\vDash$).
Updates to the concrete state are according to the operational semantics, while updates to the abstract state must preserve frames ($r$) and increase the world.
When a terminal configuration is reached, the abstract state must satisfy the postcondition.

The semantics of entailment in the logic is defined as:
\begin{align*}
&\Gamma \mid \Theta \models S \eqdef
  \forall w\in\WEL.\forall \rho \in \semel{\Gamma}.~\semel{\Gamma \vd \Theta}_\rho(w) \leq \semel{\Gamma \vd S : \ty{Spec}}_\rho(w)
\\
&\Gamma \mid \Theta \mid P \modelsel Q \eqdef
  \forall w\in\WEL.\forall \rho \in \semel{\Gamma}.~\semel{\Gamma \vd \Theta}_\rho(w) = \top\\
&\qquad\qquad\qquad\qquad \implies
    \semel{\Gamma \vd P : \ty{Prop}}_\rho(w) \subseteq \semel{\Gamma \vd Q : \ty{Prop}}_\rho(w)\!
\end{align*}
where $\semel{\Gamma \vd \Theta}_\rho(w) = \bigwedge_{T \in \Theta} \semel{\Gamma \vd T : \ty{Spec}}_\rho(w)$ (here $\bigwedge$ is lub  in $\{\bot,\top\}$).

Soundness of the logic is proved by induction on the structure of derivations (cf.\ Appendix~\ref{sec:elsound}). Using soundness, we can relate a proof of a triple $\triple{P}{e}{Q}$
to the store and trace obtained from the reduction of $e$.

\begin{theorem}[Soundness]\label{thm:soundel}
  If $\Gamma \mid \Theta \vd S$ then $\Gamma \mid \Theta \models S$.
\end{theorem}

\begin{corollary}\label{lem:el-multistep-exec}
Suppose $\Gamma \mid - \vd \triple{P}{e}{Q}$
and let $w \in \WEL, \rho \in \semel{\Gamma}$ and $m \in \semel{\Gamma \vdash P}_\rho(w)$.
Then, for all $r, t, s$ such that $t, s \vDash_w m$ and 
$\rho(e),s \xrightarrow{t'}{\!}^*\,\,  e',s' \not\rightarrow$, for some  $s',t'$, we have that
$e' \in \VAL$ and $\exists w' \geq w.\,\exists m' \in \semel{\Gamma \vdash Q}_\rho(e')(w').\,(t \cdot t'), s' \vDash_{w'} m'$.
\end{corollary}



\section{Conclusions} \label{sec:conclusions}

In this paper we demonstrated a formal approach for relating library specifications, expressed in 
separation logic, with the trace properties they enforce on the interaction between clients verified
against the specified library and the library itself. 
The distinctive strength of our technique is that it is based purely on the abstract library specification
and is independent of both client and library implementations. 
As such, it differs from the standard
verification approach where one verifies a program against a given specification  expressing the desired
property. 
%
Since our main goal has been to establish a theoretical foundation relating  specifications and trace 
properties, we focused on expressiveness rather than automation.


\subsection{Related Work}

Several lines of work have targeted static verification of safety trace properties of object-oriented and higher-order programs. 

A particularly influential approach has been Typestates~\cite{Strom:1986:typestate,Deline:2004:typestate-obj}. 
These can be seen as specifying trace
properties using pre/post-conditions. 
They have been used
to check safety temporal properties of programs, by associating abstract states to objects, 
then specifying which methods can be called at each state and how they make the state evolve.
In~\cite{Bierhoff:2007:modtc-aliasedobj,Bierhoff:2009:apicheck-ap} they are combined with aliasing information to give a sound and modular analysis of API usage protocols, and to check a specification for the iterator module similar to ours, yet based on a single object.
Multi-object properties, like the iterator one, can be captured by an extension of typestates using tracematches to specify intensional properties~\cite{naeem2008extending}.

Static analyses based on type systems have been widely used to check resource usage, like our file module example.
Linear types are used in~\cite{DeLine:2001:hlprot-llsoft} to develop Vault, a programming language used to design device drivers,
where resource management protocols can be specified explicitly using annotations in the source code.
An automatic analysis has then been developed in~\cite{Igarashi:2002:rua}.  

Type and effect systems have been used in~\cite{Skalka:2008:tteff-ho} to infer resource usage, represented by an LTS, and
combined with model checking to verify trace properties of programs.
Such systems have been applied to Featherweight Java in~\cite{Skalka:2008:tteff-oo}, where challenges coming from object orientation
like inheritance and dynamic dispatch are tackled.

Higher-order model checking has been used to provide a sound and complete resource usage analysis,
using higher-order recursion scheme model checking for a fragment of the $\mu$-calculus~\cite{Kobayashi:2009:hors-verif}.

Those approaches have been designed to support automated static verification and thus
trade off expressiveness for automation. Here we have made the opposite trade-off and focused on being able to capture expressive trace properties. 
As seen, we can specify 
non-regular properties~(\ref{sec:wb-example}, the language is visibly pushdown~\cite{AlurM04}), others that rely on tracking an unbounded number of 
objects (\ref{sec:iterator-example} and Appendix~\ref{sanitize-example}, where we need to track all valid iterators and strings respectively), and we can even go beyond 
context-free languages (\ref{stackmap-example}, the language requires an order-2 pushdown automaton~\cite{maslov1976multilevel}). On an orthogonal direction,  working with a logic with quantification we can specify traces from infinite alphabets of trace events~(\ref{sec:iterator-example}, \ref{sec:wb-example}, \ref{stackmap-example} and Appendix~\ref{sanitize-example}).

Our key contribution is a technique for formally relating separation logic
  specifications with the temporal properties they enforce,
through wrapping abstract resources with assertions about traces. 
On the other hand, a very active line of work has targeted the verification of fine-grained concurrent data structures using
program logics with histories~\cite{gotsman-rcu,feng} and separation logics~\cite{fcsl,tada}. Fine-grained concurrent
algorithms use locking at the level of individual memory operations and their correctness often relies on subtle
temporal properties about internal interactions within libraries. These approaches include primitives for reasoning explicitly about traces~\cite{gotsman-rcu,feng} or assign to (fine-grained) separation logic primitives history-oriented interpretations~\cite{fcsl}. 
Temporal reasoning is thus achieved in a different way than herein, namely by specifying and verifying
  the underlying libraries themselves, whereas 
we derive temporal properties from specifications
  that are not themselves temporal.

The F7 and F$^{*}$ programming languages provide another technique for reasoning about trace properties~\cite{f-star,f-seven}. 
The technique is primarily
aimed at verifying cryptographic primitives, but has also been applied to access control policies about
interactions between a client and resources managed through libraries~\cite{triples-hoare}. 
It is based on extending the base programming
language with a primitive for assuming that a given formula holds and an assert primitive that fails if a given formula
does not follow from all previously assumed formulas. 
Access control policies are encoded by inserting appropriate 
assume and assert statements and proving that no assert can fail. 
In comparison to our work, the approach
does not establish a formal connection between the inserted assume/assert statements and the property enforced
on the execution. It is also non-local in that any assume statement can introduce a contradiction and
break adequacy.

\subsection{Further Directions}

We demonstrated the power of even a basic higher-order separation logic 
for enforcing elaborate trace properties. A more expressive logic with better
support for shared-resource  reasoning would allow us to enforce even more elaborate
protocols. In the future we therefore intend to apply our technique to a fully 
featured concurrent higher-order separation logic~\cite{Iris}.

In the concurrent setting, linearisability~\cite{Linearizability} is a trace property that is commonly used to specify that operations of a library behave as if they were atomic.
Several concurrent separation logics adopt a different approach to specifying atomicity, due to Jacobs and Piessens~\cite{JacobsPOPL2011}.
Although there is a strong intuitive argument that this approach implies linearisability, a formal connection has not been made.
Our technique could be used to formalise such a connection.

Another direction concerns the description of the captured protocols via formal
language tools, such as automata. 
Note that such a step is not to be taken lightly, as one would need to establish a formal link between the 
description and the logical specification. A related area for future work would be a formal 
(and, ideally, automated) procedure for deriving separation logic specifications for enforcing a given 
protocol from a formal definition of the protocol.


\newpage

\bibliographystyle{splncs03}
\bibliography{references}

\begin{thebibliography}{10}
\providecommand{\url}[1]{\texttt{#1}}
\providecommand{\urlprefix}{URL }

\bibitem{AlurM04}
Alur, R., Madhusudan, P.: Visibly pushdown languages. In: Proceedings of the
  36th Annual {ACM} Symposium on Theory of Computing, Chicago, IL, USA, June
  13-16, 2004. pp. 202--211 (2004)

\bibitem{f-seven}
Bengtson, J., Bhargavan, K., Fournet, C., Gordon, A.D., Maffeis, S.: Refinement
  {T}ypes for {S}ecure {I}mplementations. ACM Trans. Program. Lang. Syst.
  33(2),  8:1--8:45 (Feb 2011)

\bibitem{Bierhoff:2007:modtc-aliasedobj}
Bierhoff, K., Aldrich, J.: Modular typestate checking of aliased objects. In:
  Proceedings of the 22nd Annual ACM SIGPLAN Conference on Object-oriented
  Programming Systems and Applications. pp. 301--320. OOPSLA'07, ACM (2007)

\bibitem{Bierhoff:2009:apicheck-ap}
Bierhoff, K., Beckman, N.E., Aldrich, J.: Practical {API} protocol checking
  with access permissions. In: Proceedings of the 23rd European Conference on
  Object-Oriented Programming. pp. 195--219. ECOOP'09, Springer-Verlag (2009)

\bibitem{Biering:2007:bi-hyp-hosl}
Biering, B., Birkedal, L., Torp-Smith, N.: Bi-hyperdoctrines, higher-order
  separation logic, and abstraction. ACM Trans. Program. Lang. Syst.  29(5)
  (Aug 2007)

\bibitem{triples-hoare}
Borgstr{\"o}m, J., Gordon, A.D., Pucella, R.: Roles, {S}tacks, {H}istories: A
  {T}riple for {H}oare. Journal of Functional Programming  21,  159--207 (2011)

\bibitem{DeLine:2001:hlprot-llsoft}
DeLine, R., F\"{a}hndrich, M.: Enforcing high-level protocols in low-level
  software. In: Proceedings of the ACM SIGPLAN 2001 Conference on Programming
  Language Design and Implementation. pp. 59--69. PLDI '01, ACM (2001)

\bibitem{Deline:2004:typestate-obj}
DeLine, R., F{\"a}hndrich, M.: Typestates for objects. In: Object-Oriented
  Programming: Proceedings of the 18th European Conference on Object-Oriented
  Programming. p. 465. ECOOP'04, Springer-Verlag (2004)

\bibitem{feng}
Fu, M., Li, Y., Feng, X., Shao, Z., Zhang, Y.: Reasoning about {O}ptimistic
  {C}oncurrency {U}sing a {P}rogram {L}ogic for {H}istory. In: In Proceedings
  of CONCUR (2010)

\bibitem{POSIXSL}
Gardner, P., Ntzik, G., Wright, A.: Local reasoning for the {POSIX} file
  system. In: ESOP. pp. 169--188 (2014)

\bibitem{gotsman-rcu}
Gotsman, A., Rinetzky, N., Yang, H.: Verifying concurrent memory reclamation
  algorithms with grace. In: Proceedings of ESOP (2013)

\bibitem{Grigore:2013:runtimeverif-regaut}
Grigore, R., Distefano, D., Petersen, R.L., Tzevelekos, N.: Runtime
  verification based on register automata. In: Proceedings of the 19th
  International Conference on Tools and Algorithms for the Construction and
  Analysis of Systems. pp. 260--276. TACAS'13, Springer-Verlag (2013)

\bibitem{Linearizability}
Herlihy, M.P., Wing, J.M.: {Linearizability: a correctness condition for
  concurrent objects}. ACM Trans. Program. Lang. Syst.  12(3),  463--492 (Jul
  1990)

\bibitem{Igarashi:2002:rua}
Igarashi, A., Kobayashi, N.: Resource usage analysis. In: Proceedings of the
  29th ACM SIGPLAN-SIGACT Symposium on Principles of Programming Languages. pp.
  331--342. POPL '02, ACM (2002)

\bibitem{Ishtiaq:2001:bi-asertlog}
Ishtiaq, S.S., O'Hearn, P.W.: Bi as an assertion language for mutable data
  structures. In: Proceedings of the 28th ACM SIGPLAN-SIGACT Symposium on
  Principles of Programming Languages. pp. 14--26. POPL'01, ACM (2001)

\bibitem{JacobsPOPL2011}
Jacobs, B., Piessens, F.: Expressive modular fine-grained concurrency
  specification. In: POPL. pp. 271--282 (2011)

\bibitem{Iris}
Jung, R., Swasey, D., Sieczkowski, F., Svendsen, K., Turon, A., Birkedal, L.,
  Dreyer, D.: Iris: {M}onoids and {I}nvariants as an {O}rthogonal {B}asis for
  {C}oncurrent {R}easoning. In: Proceedings of POPL (2015)

\bibitem{Kobayashi:2009:hors-verif}
Kobayashi, N.: Types and higher-order recursion schemes for verification of
  higher-order programs. In: Proceedings of the 36th Annual ACM SIGPLAN-SIGACT
  Symposium on Principles of Programming Languages. pp. 416--428. POPL '09, ACM
  (2009)

\bibitem{Krishnaswami:2009:dp-sl}
Krishnaswami, N.R., Aldrich, J., Birkedal, L., Svendsen, K., Buisse, A.: Design
  patterns in separation logic. In: Proceedings of the 4th International
  Workshop on Types in Language Design and Implementation. pp. 105--116.
  TLDI'09, ACM (2009)

\bibitem{Laird:2007:fulabstr-tracesem}
Laird, J.: A fully abstract trace semantics for general references. In:
  Proceedings of the 34th International Conference on Automata, Languages and
  Programming. pp. 667--679. ICALP'07, Springer-Verlag (2007)

\bibitem{maslov1976multilevel}
Maslov, A.N.: Multilevel stack automata. Problems Inform. Transmission  12(1),
  38--42 (1976)

\bibitem{naeem2008extending}
Naeem, N.A., Lhot{\'a}k, O.: Extending typestate analysis to multiple
  interacting objects. OOPSLA’08: Proceedings of Object-Oriented Programming,
  Systems, Languages and Applications  (2008)

\bibitem{Parkinson:2005:seplog-abstr}
Parkinson, M., Bierman, G.: Separation logic and abstraction. In: Proceedings
  of the 32nd ACM SIGPLAN-SIGACT Symposium on Principles of Programming
  Languages. pp. 247--258. POPL'05, ACM (2005)

\bibitem{Reynolds:2002:separation-logic}
Reynolds, J.C.: Separation logic: {A} logic for shared mutable data structures.
  In: Proceedings of the 17th Annual IEEE Symposium on Logic in Computer
  Science. pp. 55--74. LICS'02, IEEE (2002)

\bibitem{tada}
da~Rocha~Pinto, P., Dinsdale{-}Young, T., Gardner, P.: {TaDA}: {A} {L}ogic for
  {T}ime and {D}ata {A}bstraction. In: Proceedings of ECOOP (2014)

\bibitem{fcsl}
Sergey, I., Nanevski, A., Banerjee, A.: Specifying and {V}erifying {C}oncurrent
  {A}lgorithms with {H}istories and {S}ubjectivity. In: Proceedings of ESOP
  (2015)

\bibitem{Skalka:2008:tteff-oo}
Skalka, C.: Types and trace effects for object orientation. Higher-Order and
  Symbolic Computation  21(3),  239--282 (2008)

\bibitem{Skalka:2008:tteff-ho}
Skalka, C., Smith, S., Van~Horn, D.: Types and trace effects of higher order
  programs. Journal of Functional Programming  18(02),  179--249 (2008)

\bibitem{Strom:1986:typestate}
Strom, R.E., Yemini, S.: Typestate: A programming language concept for
  enhancing software reliability. IEEE Trans. Softw. Eng.  12(1),  157--171
  (Jan 1986)

\bibitem{f-star}
Swamy, N., Chen, J., Fournet, C., Strub, P.Y., Bhargavan, K., Yang, J.: Secure
  {D}istributed {P}rogramming with {V}alue-{D}ependent {T}ypes. Journal of
  Functional Programming  23(4),  402--451 (2013)

\end{thebibliography}

\newpage
\appendix
\section{String Sanitisation}\label{sanitize-example}

Another interesting example is string sanitisation~\cite{Grigore:2013:runtimeverif-regaut}. Suppose strings are fetched from a web form,
processed internally and passed as part of an SQL query to a database server. To avoid injection attacks,
all inputs from the web form should be sanitized before being passed to the database server. We can express
this as a \emph{taint protocol}. Strings received from the user are considered \emph{tainted} initially.
Taint remains with a string and is also passed to any string produced by processing some tainted string. 
The only way to remove taint from a string is by {\em sanitising} it. 
A protocol we may require is:
\begin{quote}\em
No tainted string can reach the database server.
\end{quote}
To express this trace property formally, we first define a suitable wrapping around a string library, 
that emits events describing the interaction between library and client. Suppose that a string library consists
of five methods: an \tty{input} method for obtaining strings from the user via a web form, a \tty{constant}
method for declaring string constants in the code, a \tty{sanitize} method, a concatenation method, \tty{concat},
and a sink method, \tty{sink}, for sending a given string to the database server. Formally, we represent the
library as a 5-tuple of these operations. We can now define a wrapper that takes a string library and returns
an instrumented string library. The instrumentation is fairly straightforward and simply emits a suitable event
indicating the operation called and the argument and/or return value of the given operation, when relevant:
\begin{align*}
&\wrapstr \eqdef\ \lambda (\mathit{input}, \mathit{constant}, \mathit{sanitize}, \mathit{concat}, \mathit{sink}).~(\\
  &\quad
  \begin{aligned}[t]
  &
    \lambda \_.~\TEMPletin{r = \mathit{input}()}{\TEMPemit \pair{\tty{input}}{r};\,r}, \\
  &
    \lambda y.~\TEMPletin{r = \mathit{constant}(y)}{\TEMPemit \pair{\tty{constant}}{r};\,r}, \\
  &
    \lambda y.~\mathit{sanitize}(y);\,\TEMPemit \pair{\tty{sanitize}}{y}, \\
  &
    \lambda (y_1, y_2).~\TEMPletin{r = \mathit{concat}(y_1, y_2)}{\TEMPemit \pair{\tty{concat}}{r, y_1, y_2};\,r}, \\
  &
    \lambda y.~\mathit{sink}(y);\,\TEMPemit \pair{\tty{sink}}{y}\,)
  \end{aligned}
\end{align*}
We can now define the trace property as a constraint on the traces generated by linking a client
with an instrumented library. 
We let the alphabet be the following countable set:
\begin{align*}
\Sigma &= \{\pair{e}{s} \mid e \in \{\tty{constant},\tty{input},\tty{sanitize},\tty{sink}\}, s \in \mathit{Loc}\} \\
&\quad\,\cup\{\pair{\tty{concat}}{s,s_1,s_2}\mid s,s_1,s_2\in\mathit{Loc}\}
\end{align*}
Intuitively, we want to constrain traces such that for any 
$\langle \tty{sink}, s \rangle$ event in the trace, all the strings $s'$ used to construct $s$ are
safe to emit, meaning that the trace contains a corresponding \tty{constant} event or \tty{sanitize}
event for such $s'$. We first define, by induction on $t$, a predicate $\esafep(s, t)$ to express that the string $s$
is safe to emit given that the current trace is $t$. We let
$\esafep(s, \varepsilon) \eqdef \bot$ and:
\[\begin{array}{rl}
\esafep(s, h::t) \eqdef & \esafep(s, t) \lor h = \langle \tty{constant}, s \rangle \lor h = \langle \tty{sanitize}, s \rangle~\lor \\
  & (\exists s_1, s_2 \ldotp h = \langle \tty{concat}, s, s_1, s_2 \rangle \land \esafep(s_1, t) \land \esafep(s_2, t))
  \end{array}\]
Note such inductive predicates are definable in higher-order logic.

We wish to ensure that once a string has been sanitised it can never become tainted again or,
stated in terms of $\esafep$, once a string is safe for $t$ then it is also safe for all future
histories $t'$ such that $t \leqprefix t'$. To ensure this we require the library to ensure
that \tty{input}, \tty{constant} and \tty{concat} always return fresh string pointers that
have never been used before. We thus define the language of valid traces, $\Lstr$, as the set of all
strings $t \in \Sigma^*$ such that $\validsanip(t)$, where:
\begin{align*}
\validsanip(t) \eqdef
    (&\forall n, s \ldotp t[n] = \langle \tty{sink}, s \rangle \implies \esafep(s, t)) \lor \notfreshp(t)
 \\
\isallocp(s, t, n) &\eqdef t[n] = \langle \tty{constant}, s \rangle \lor t[n] = \langle \tty{input}, s \rangle \lor{}\\
&\quad\;\  \exists s_1, s_2 \ldotp t[n] = \langle \tty{concat}, s, s_1, s_2 \rangle \\
\notfreshp(t) &\eqdef
  \exists n, m, s \ldotp n \neq m \land \isallocp(s, t, n) \land \isallocp(s, t, m)
\end{align*}
The $\validsanip(t)$ predicate expresses that either every sink event in $t$ uses a safe
string $s$ or the library does not ensure sufficient freshness of string pointers. In the latter case
 we do not constrain
the client at all. We thus have the following crucial property, which ensures that if a string
$s$ is safe for the current trace $t$ then it is also safe for any future trace $t'$ with
sufficiently fresh string pointers:
\begin{align*}
\forall s, t_1, t_2 \ldotp
  t_1 \leqprefix t_2 \land \esafep(s, t_1) &\implies \esafep(s, t_2) \lor \notfreshp(t_2)
\end{align*}

Below we define a separation logic specification for the string sanitisation library that
ensures clients only call \tty{sink} with safe strings. The specification uses a resource
$\Rp$ for the local state of the sanitisation library, a string resource $\strp(x)$ for
each string and a safe resource $\safep(x)$ which expresses that the string $x$ is safe.
As expected, \tty{sanitize} turns a string into a safe string and \tty{concat} applied to
safe strings yields a safe string.  \vspace{-1mm}
\begin{align*}
&
  \Phistr(P_{0}, (\mathit{input}, \mathit{constant}, \mathit{sanitize}, \mathit{concat}, \mathit{sink})) \eqdef \\
&\;\;
  \begin{aligned}[t]
  &
    \exists \Rp : \ty{Prop} \ldotp \exists \strp, \safep : \ty{Val} \rightarrow \ty{Prop} \ldotp  \\
  &\quad
    \valid(P_{0}\! \implies\! \Rp) \land 
    \valid(\forall s \ldotp \safep(s)\! \implies\! \safep(s) * \safep(s))~\land \\
  &\quad
    \triple{\Rp}{\mathit{input}()}{r \ldotp \Rp * \strp(r)}~\land 
    \forall s \ldotp 
    \triple{\Rp}{\mathit{constant}(s)}{ r \ldotp \Rp * \strp(r) * \safep(r)}~\land \\
  &\quad
    \forall s \ldotp 
    \triple{\Rp * \strp(r)}{\mathit{sanitize}(s)}{\Rp * \strp(r) * \safep(s)}~\land \\
  &\quad
    \forall s_1, s_2 \ldotp 
    \begin{aligned}[t]
    & \{\Rp * \strp(s_1) * \strp(s_2)\} 
     \quad {\mathit{concat}(s_1, s_2)} 
    \{r \ldotp \Rp * \strp(s_1) * \strp(s_2) * \strp(r)\}~\land 
    \end{aligned} \\
  &\quad
    \forall s_1, s_2 \ldotp 
    \begin{aligned}[t]
    & \{ \Rp * \strp(s_1) * \strp(s_2) * \safep(s_1) * \safep(s_2) \} \\
    &\quad
      \mathit{concat}(s_1, s_2) \{ x \ldotp \Rp * \strp(s_1) * \strp(s_2) * \strp(x) * \safep(x) \}~\land
    \end{aligned} \\
  &\quad
    \forall s \ldotp 
    \triple{\Rp * \strp(s) * \safep(s)}{\mathit{sink}(s)}{\Rp * \strp(s)}
  \end{aligned}
\end{align*}
As usual, we proceed by proving that the instrumentation preserves the library specification
and generates traces in $\Lstr$.
\begin{lemma}
$P_0, ops \mid \Phistr(P_0, ops) \vdel \Phistr(P_0 * \trace(\varepsilon) * \inv(\Lstr), \wrapstr(ops))$.
\end{lemma}
\begin{proof}[Proof sketch]
To prove that the instrumentation preserves the library specification we first have
to define instrumented versions of the abstract representation predicates. We let the
local state resource $\Rpw$ take ownership of the $\trace$ resource and let $\safepw$
assert the existence of some trace prefix in which the given string is safe or for
which the library failed to ensure sufficiently fresh string pointers:
\begin{align*}
\Rpw &\eqdef \exists t \in \Lstr \ldotp \Rp * \trace(t) * \inv(\Lstr),\qquad\quad \strpw(s) \eqdef \strp(s) \\
\safepw(s) &\eqdef \exists t \in \Lstr \ldotp \safep(s) * \hist(t) * (\esafep(s, t) \lor \notfreshp(t))
\end{align*}
It remains to show that the instrumentation preserves the specifications of each of the library
methods. 
The most interesting cases are \tty{sink}, and \tty{concat} when on
two safe strings. For the latter we have:\vspace{-1mm}
\begin{source}
[Context $s_1, s_2 : \ty{Val}$]
$\{ \Rpw * \strpw(s_1) * \strpw(s_2) * \safepw(s_1) * \safepw(s_2) \}$
$\{ \Rp * \strp(s_1) * \strp(s_2) * \safep(s_1) * \safep(s_2) * \exists t, t_1, t_2 \in \Lstr \ldotp \trace(t) * \inv(\Lstr)$ 
  $* \hist(t_1) * \hist(t_2) * (\esafep(s_1, t_1) \lor \notfreshp(t_1)) * (\esafep(s_2, t_2) \lor \notfreshp(t_2)) \}$
let $s$ = concat($s_1$, $s_2$) in
$\{ \Rp * \strp(s_1) * \strp(s_2) * \strp(s) * \safep(s) * \exists t, t_1, t_2 \in \Lstr \ldotp \trace(t) * \inv(\Lstr)$
  $*\hist(t_1) * \hist(t_2) * (\esafep(s_1, t_1) \lor \notfreshp(t_1)) * (\esafep(s_2, t_2) \lor \notfreshp(t_2)) \}$
emit<<concat, $s$, $s_1$, $s_2$>>;
$\{ \Rp * \strp(s_1) * \strp(s_2) * \safep(s) * \exists t, t' \in \Lstr \ldotp$
  $\quad  \trace(t) * \inv(\Lstr) * \hist(t') *
   (\esafep(s, t')  \lor \notfreshp(t')) \}$
$\{ \Rpw * \strpw(s_1) * \strpw(s_2) * \strpw(s) * \safepw(s) \}$
\end{source}\vspace{-.5mm}
here we use the following properties of $\esafep$ to verify the emit:
\begin{align*}
\forall s, s_1, s_2, t \ldotp 
  \esafep(s_1, t) \land{}& \esafep(s_2, t) \implies \esafep(s, t \cdot \langle \tty{concat}, s, s_1, s_2 \rangle) 
\\
\forall s, s_1, s_2, t \ldotp
  \esafep(s_1, t) \land{}& \esafep(s_2, t)\implies t \cdot \langle \tty{concat}, s, s_1, s_2 \rangle \in \Lstr\\[-7mm]
\end{align*}
For the sink case, we have:\vspace{-1mm}
\begin{source}
[Context $s : \ty{Val}$]
$\{ \Rpw * \strpw(s) * \safepw(s) \}$
$\{ \Rp * \strp(s) * \safep(s) *\exists t, t' \in \Lstr \ldotp \trace(t) * \inv(\Lstr) * \safep(s) * \hist(t') * \esafep(s, t') \}$
sink($s$);
$\{ \Rp * \strp(s) * \exists t, t' \in \Lstr \ldotp \trace(t) * \inv(\Lstr) *\safep(s) * \hist(t') * (\esafep(s, t') \lor \notfreshp(t')) \}$
emit<<sink, $s$>>;
$\{ \Rp * \strp(s) * \exists t \in \Lstr \ldotp \trace(t \cdot \langle \tty{sink}, s \rangle) * \inv(\Lstr) \}$
$\{ \Rpw * \strpw(s) \}$
\end{source}
\end{proof}



\section{Soundness}
\label{sec:elsound}

\begin{figure*}
\begin{align*}
  \semel{\Gamma \vdash x : \tau}_\rho &\eqdef  \rho(x) \\
  \semel{\Gamma \vdash \bot : \ty{Prop}}_\rho(w) &\eqdef  \emptyset \\
  \semel{\Gamma \vdash \top : \ty{Prop}}_\rho(w) &\eqdef  \semel{\Gamma \vdash \emp : \ty{Prop}}_\rho(w) \eqdef  \MEL \\
  \semel{\Gamma \vdash \lambda x : \tau \ldotp M : \tau \rightarrow \sigma}_\rho &\eqdef 
    \lambda v \in \semel{\tau} \ldotp \semel{\Gamma, x : \tau \vdash M : \sigma}_{\rho[x \mapsto v]} \\
  \semel{\Gamma \vdash M~N : \sigma}_\rho &\eqdef 
    \semel{\Gamma \vdash M : \tau \rightarrow \sigma}_\rho(\semel{\Gamma \vdash N : \tau}_\rho) \\
  \semel{\Gamma \vdash \forall x : \tau \ldotp P : \ty{Prop}}_\rho(w) &\;\ \eqdef  \bigcap\nolimits_{v \in \semel{\tau}} \semel{\Gamma, x : \tau \vdash P}_{\rho{[x \mapsto v]}}(w) 
\\
  \semel{\Gamma \vdash P \land Q : \ty{Prop}}_\rho(w)  &
\eqdef  
\semel{\Gamma \vdash P:\ty{Prop}}_\rho(w)\cap\semel{\Gamma\vdash  Q : \ty{Prop}}_\rho(w)
\\
  \semel{\Gamma \vdash P \Rightarrow Q : \ty{Prop}}_\rho(w) &
\eqdef  
\Setb{m}{\begin{array}{@{}l@{}} \forall m'\geq m \ldotp \forall w' \geq w \ldotp
 m'\in \semel{\Gamma \vdash P : \ty{Prop}}_\rho(w') \\ {} \hfill \implies\, m' \in \semel{\Gamma \vdash Q : \ty{Prop}}_\rho(w') \end{array}} 
\\
  \semel{\Gamma \vdash P * Q : \ty{Prop}}_\rho(w) &
\eqdef  
\Setb{m}{
\begin{array}{@{}l} 
\exists m_1, m_2 \ldotp
m_1 \in \semel{\Gamma \vdash P : \ty{Prop}}_\rho(w) \\ {}
\land m_2 \in \semel{\Gamma \vdash Q : \ty{Prop}}_\rho(w) 
\land m=m_1 \bullet m_2 
\end{array}} 
\\
  \semel{\Gamma \vdash P \wand Q : \ty{Prop}}_\rho(w) &
\eqdef  
\Setb{m}{%
    \begin{array}{@{}l}
        \forall w' \geq w \ldotp \forall m' \geq m \ldotp \forall m'' \in \MEL \ldotp \\
        m' \bullet m''~\text{defined}\land m'' \in \semel{\Gamma \vdash P : \ty{Prop}}_\rho(w') \\
{}\hfill \implies  m' \bullet m'' \in \semel{\Gamma \vdash Q : \ty{Prop}}_\rho(w')
    \end{array}}
\\
  \semel{\Gamma \vdash M \mapsto N : \ty{Prop}}_\rho(w) &
\eqdef  
    \{\,m\in\MEL\mid m(\semel{\Gamma \vdash M : \ty{Loc}}_\rho) = \semel{\Gamma \vdash N : \ty{Val}}_\rho \}
\\
  \semel{\Gamma \vdash \spec(M) : \ty{Prop}}_\rho(w) &
\eqdef  
    \{\,m\in\MEL\mid\semel{\Gamma \vdash M : \ty{Spec}}_\rho(w) = \top\,\}
\\
  \semel{\Gamma \vdash M=_\tau N : \ty{Spec}}_\rho(w) &
\eqdef  
    (\,\semel{\Gamma \vdash M : \tau}_\rho = \semel{\Gamma \vdash N : \tau}_\rho\,)
\\
  \semel{\Gamma \vdash \valid(P) : \ty{Spec}}_\rho(w) &
\eqdef  
    (\,\semel{\Gamma \vdash P : \ty{Prop}}_\rho(w) = \MEL\,)
\\
  \semel{\Gamma\vdash \trace(t) : \ty{Prop}}_\rho(w) &
\eqdef  
\Setb{ (h,\tau) }{ \tau \geq (\topTrace,\semel{\Gamma \vdash t : \ty{seq}~\ty{Val}}_\rho) }
\\
  \semel{\Gamma\vdash \hist(t) : \ty{Prop}}_\rho(w) &
\eqdef  
\Setb{ (h,\tau) }{ \tau \geq (\botHist,\semel{\Gamma \vdash t : \ty{seq}~\ty{Val}}_\rho) } 
\\
  \semel{\Gamma\vdash \inv(I) : \ty{Prop}}_\rho(w) &
\eqdef  
\Setb{ m\! }{ \!\pi_2(w) = \Setb{t\!}{\!\semel{\Gamma \vdash I : \ty{seq}~\ty{Val} \to \ty{Bool}}_\rho (t) = \ty{true}} \!} 
\\
  \semel{\Gamma \vdash \triple{P}{e}{Q} : \ty{Spec}}_\rho(w) &
\eqdef  \begin{array}{@{}l@{}}
\forall w' \geq w \ldotp \semel{\Gamma \vdash P : \ty{Prop}}_\rho(w') \\
{} \hspace{4em} \subseteq \mathsf{wp}(\semel{\Gamma \vdash Q : \ty{Val} \to \ty{Prop}}_\rho)(\rho(e))(w')
\end{array}
\end{align*}
\[
  \mathsf{wp}(Q) \eqdef \nu \mathit{wp}' \ldotp \lambda e, w \ldotp
    \Setb{m}{ \begin{array}{@{}l@{}}
        \forall r, t, s \ldotp~ t, s \vDash_w m \bullet r \implies  \\
                \ ( e,s \nrightarrow {} \implies e \in \VAL \land m \in Q(e)(w)) \\
                \  {} \land \forall a, e', s' \ldotp~ e, s \xrightarrow{a} e', s' 
                \implies \\ \ \ \ \exists w' \geq w, m' \ldotp (t \cdot a), s' \vDash_{w'} m' \bullet r \land m' \in \mathit{wp}'(e')(w')
              \end{array}}
\]
\begin{align*}
\text{where}\qquad  t, (h, \gamma) \vDash_{(\gamma', I)} (h', (\topTrace, t')) &\iff t = t' \land h = h' \land \gamma = \gamma' \land t \in I \\
  t, (h, \gamma) \vDash_{(\gamma', I)} (h', (\botHist, t')) &\iff t' \leqprefix t \land h = h' \land \gamma = \gamma' \land t \in I
\end{align*}
\caption{Semantics of terms (all cases).}\label{fig:semELfull}
\end{figure*}

The semantics of all terms is given in Figure~\ref{fig:semELfull}.

\begingroup
\def\thetheorem{\ref{thm:soundel}}
\begin{theorem}[Soundess]
  If $\Gamma \mid \Theta \vd S$ then $\Gamma \mid \Theta \modelsel S$.
\end{theorem}
\addtocounter{theorem}{-1}
\endgroup

Soundness is established by showing that each of the proof rules is semantically valid.

As a notational convenience, given $f, g : \EXP \to \semel{\ty{Prop}}$ we write $f \sqsubseteq g$ just if $\forall e, w \ldotp f(e)(w) \subseteq g(e)(w)$. We moreover define $f\sqcup g$ as $\lambda e,w.\,f(e)(w)\cup g(e)(w)$, and extend this notation to arbitrary unions $\bigsqcup S$ for $S\subseteq \EXP \to \semel{\ty{Prop}}$.

Moreover, for each $p,q\in\semel{\ty{Prop}}$, we let $p*q$ stand for $\lambda w.\,p(w)* q(w)$; whereby,
for each $X,Y\subseteq\MEL$, $X*Y=\{m\in\MEL\mid \exists m_1,m_2.~m=m_1\bullet m_2\land m_1\in X\land m_2\in Y\}$.

\begin{lemma}[\textsc{Hyp} rule]
  \[
    \Gamma \mid \Theta, S \modelsel S
  \]
\end{lemma}
\begin{proof}
  Let $\rho \in \semel{\Gamma}$.
  Since $\forall w\in\WSL.\bigwedge_{T \in (\Theta, S)} \semel{\Gamma \vdash T : \ty{Spec}}_\rho(w) \leq \semel{\Gamma \vdash S : \ty{Spec}}_\rho(w)$, we have $\Gamma \mid \Theta, S \modelsel S$ by definition.
\end{proof}

\begin{lemma}[\textsc{Ret} rule]
  \[
    \Gamma, v : \ty{Val} \mid \Theta \modelsel \triple{\top}{v}{r \ldotp r = v}
  \]
\end{lemma}
\begin{proof}
  Let $\rho \in \semel{\Gamma, v : \ty{Val}}$ and $w \in \WEL = \FENV \times \mathit{Trace}$.
  \begin{align*}
&    \mathsf{wp}(\semel{\Gamma, v : \ty{Val} \vd \lambda r \ldotp r = v : \ty{Val} \to \ty{Prop}}_\rho)(\rho(v))(w)  \\
              &= \Setb{ m \in \MEL }{m \in \semel{\Gamma, v : \ty{Val} \vd \lambda r \ldotp r = v : \ty{Val} \to \ty{Prop}}_\rho(w)(\rho(v)) } \\
              &= \MEL \\
              &\supseteq \semel{\Gamma, v : \ty{Val} \vd \top : \ty{Prop}}_\rho(w)
  \end{align*}
  Hence $\Gamma \mid \Theta \modelsel \triple{\top}{v}{r \ldotp r = v}$.
\end{proof}

\begin{lemma}[Monotonicity of $\mathsf{wp}$] \label{lem:monoEL}
  For all $p_1, p_2 : \VAL \to \semel{\ty{Prop}}$ and $w\in\WEL$, if $\forall v,w'\geq w.\,p_1(v)(w') \subseteq p_2(v)(w')$ then 
$\forall e,w'\geq w.\,\mathsf{wp}(p_1)(e)(w') \subseteq \mathsf{wp}(p_2)(e)(w')$.
\end{lemma}
\begin{proof}
  Straightforward, by co-induction.
\end{proof}

\nt{This is straightforward even in the updated version. In case anyone asks, it sts that $f\sqsubseteq \ty{wp}(p_2)$, where
\[
f(e)(w')\eqdef\begin{cases}
\ty{wp}(p_1)(e)(w') & \text{if }w'\geq w\\ \bot & \text{otherwise}
\end{cases}\]
and for the latter it sts that $f\sqsubseteq H(f)$, where $H$ the endofunction from $\ty{wp}(p_2)=\ty{gfp}(H)$. And this follows from the hypothesis.
}

\begin{lemma} \label{lem:bindhelpEL}
  For all $q : \VAL\to\VAL \to \semel{\ty{Prop}}$, $e \in \EXP$, evaluation
  contexts $K$, and $w \in \WEL$,
  \[
    \mathsf{wp}(\lambda v \ldotp \mathsf{wp}(q(v))(K[v]))(e)(w) \subseteq 
    \mathsf{wp}(\bigsqcup q)(K[e])(w)
  \]
where $\bigsqcup q\eqdef \lambda u \ldotp \lambda w \ldotp \cup_{v \in \semsl{\ty{Val}}} q(v)(u)(w)$.
\end{lemma}
\begin{proof}
  By co-induction.
  Assume $m \in \mathsf{wp}(\lambda v \ldotp \mathsf{wp}(q(v))(K[v]))(e)(w)$ and that $r, t, s$ are such that $t, s \vDash_w m \bullet r$.
  If $K[e], s \nrightarrow$ then $e, s \nrightarrow$ and thus $e \in \VAL$ and $m \in \mathsf{wp}(q(e))(K[e])(w)$ as required, using continuity of $\ty{wp}$.
  Suppose that $K[e], s \xrightarrow{a} e', s'$.
  If $e \in \VAL$ then, as above, $m \in \mathsf{wp}(q(e))(K[e])(w)$.
  Otherwise, there exists $e''$ such that $e' = K[e'']$ and $e, s \xrightarrow{a} e'', s'$.
  It then follows from the assumption that there exist $w' \geq w$ and $m'$ with $(t \cdot a), s' \vDash_{w'} m' \bullet r$ and $m' \in \mathsf{wp}(\lambda v \ldotp \mathsf{wp}(q(v))(K[v]))(e'')(w')$.
  By the co-inductive assumption, $m' \in \mathsf{wp}(\bigsqcup q)(K[e''])(w')$ and so $m \in \mathsf{wp}(\bigsqcup q)(K[e])(w)$ as required.
\end{proof}

\begin{corollary}[\textsc{Bind} rule]
  If
  \begin{gather*}
    \Gamma \mid \Theta \modelsel \triple{P}{e}{x \ldotp Q} \\
    \Gamma, x : \ty{Val} \mid \Theta \modelsel \triple{Q}{K[x]}{r \ldotp R}
  \end{gather*}
  where $x \notin FV(\Theta)$, then
  \[
    \Gamma \mid \Theta \modelsel \triple{P}{ K{[e]} }{ r \ldotp \exists x : \ty{Val} \ldotp R}
  \]
\end{corollary}
\begin{proof}
  Follows from Lemmas~\ref{lem:bindhelpEL} and \ref{lem:monoEL}.
\end{proof}

\begin{corollary}[\textsc{Csq} rule]
  If
  \begin{gather*}
    \Gamma \mid \Theta \mid P_1 \modelsel P_2 \\
    \Gamma \mid \Theta \modelsel \triple{P_2}{e}{Q_2} \\
    \Gamma, x : \ty{Val} \mid \Theta \mid Q_2(x) \modelsel Q_1(x)
  \end{gather*}
  then
  \[
    \Gamma \mid \Theta \modelsel \triple{P_1}{e}{Q_1}
  \]
\end{corollary}
\begin{proof}
  Let $\rho \in \semel{\Gamma}$ and $w \in \WEL$ be such that $\semel{\Gamma \vd \Theta}_\rho (w)$.
  Suppose that $m \in \semel{\Gamma \vd P_1 : \ty{Prop}}_\rho(w)$.
  By the first assumption, it follows that $m \in \semel{\Gamma \vd P_2 : \ty{Prop}}_\rho(w)$.
  By the second assumption, it follows that $m \in \mathsf{wp}(\semel{\Gamma \vd Q_2 : \ty{Val} \to \ty{Prop}}_\rho)(\rho(e))(w)$.
  By the third assumption, we have $\semel{\Gamma \vd Q_2 : \ty{Val} \to \ty{Prop}}_\rho \sqsubseteq \semel{\Gamma \vd Q_1 : \ty{Val} \to \ty{Prop}}_\rho$.
  By Lemma~\ref{lem:monoEL}, it then follows that $m \in \mathsf{wp}(\semel{\Gamma \vd Q_1 : \ty{Val} \to \ty{Prop}}_\rho)(\rho(e))(w)$, as required.
\end{proof}

\begin{lemma}[Framing] \label{lem:framingEL}
  For all $p : \semel{\ty{Prop}}$, $q : \semel{\ty{Val}} \to \semel{\ty{Prop}}$, $e \in \EXP$ and $w \in \WEL$,
  \[
    (p * \mathsf{wp}(q)(e))(w) \subseteq \mathsf{wp}(\lambda v \ldotp p * q(v))(e)(w)
  \]
\end{lemma}
\begin{proof}
  By co-induction.
  It suffices to show that, for all $e, w$, $(p * \mathsf{wp}(q)(e))(w) \subseteq H(\lambda e' \ldotp p * \mathsf{wp}(q)(e'))(e)(w)$ where $H$ is such that $\mathsf{wp}(\lambda v \ldotp p * q(v)) = \nu \mathit{wp}' \ldotp H(\mathit{wp}')$ (as implied in the definition of $\mathsf{wp}$).
  Take any $w$, $m \in (p * \mathsf{wp}(q)(e))(w)$, $r, t, s$ with $s \vDash_w m \bullet r$.
  Then there exist $m_1, m_2$ such that $m' = m_1 \bullet m_2$, $m_1 \in p(w)$ and $m_2 \in \mathsf{wp}(q)(e)(w)$.
  If $e, s \nrightarrow$ then, since $t, s \vDash_w m_2 \bullet (m_1 \bullet r)$, it follows that $e \in \VAL$ and $m_2 \in q(e)(w)$; we can hence conclude that $m \in H(\lambda e' \ldotp p * \mathsf{wp}(q)(e'))(e)(w)$.
  If $e, s \xrightarrow{a} e', s'$ then since $m_2 \in \mathsf{wp}(q)(e)(w)$, it follows that there are $w' \geq w$ and there exists $m_2'$ such that $(t \cdot a), s' \vDash_{w'} m_2' \bullet (r \bullet m_1)$ and $m_2' \in \mathsf{wp}(q)(e')(w')$; hence $m_1 \bullet m_2' \in (p * \mathsf{wp}(q)(e'))(w')$ and we thus have that $m \in H(\lambda e' \ldotp p * \mathsf{wp}(q)(e'))(e)(w)$, as required.
\end{proof}
\begin{corollary}[\textsc{Frame} rule]
  If $\Gamma \mid \Theta \modelsel \triple{P}{e}{Q}$ then $\Gamma \mid \Theta \modelsel \triple{P * R}{e}{\lambda r \ldotp Q(r) * R}$.
\end{corollary}
\begin{proof}
  Let $\rho \in \semel{\Gamma}$ and $w \in \WEL$ be such that $\semel{\Gamma \vdash \Theta}_\rho (w)$.
  By assumption,
  \[
    \semel{\Gamma \vd P : \ty{Prop}}_\rho (w) \subseteq \mathsf{wp}(\semel{\Gamma \vd Q : \ty{Val} \to \ty{Prop}}_\rho)(\rho(e))(w)
  \]
  It follows by monotonicity of $*$ and Lemma~\ref{lem:framingEL} that
  \begin{align*}
&    \semel{\Gamma \vd P * R : \ty{Prop}}_\rho(w) \\
                                                 &= \semel{\Gamma \vd P : \ty{Prop}}_\rho(w) * \semel{\Gamma \vd R : \ty{Prop}}_\rho(w) \\
                                                 &\subseteq \mathsf{wp}(\semel{\Gamma \vd Q : \ty{Val} \to \ty{Prop}}_\rho)(\rho(e))(w) * \semel{\Gamma \vd R : \ty{Prop}}_\rho(w) \\
                                                 &\subseteq \mathsf{wp}(\lambda v \ldotp \semel{\Gamma \vd Q : \ty{Val} \to \ty{Prop}}_\rho(v) * \semel{\Gamma \vd R : \ty{Prop}}_\rho)(\rho(e))(w)
  \end{align*}
  as required.
\end{proof}

\begin{lemma}[\textsc{Abs} rule]
  If
  \[
    \Gamma, x : \ty{Val} \mid \Theta \modelsel \triple{P}{e}{Q}
  \]
  where $x \notin FV(\Theta)$, then
  \[
    \Gamma \mid \Theta \modelsel \triple{\emp}{\lambda x \ldotp e}{r \ldotp \forall x : \ty{Val} \ldotp \spec(\triple{P}{r x}{Q})}
  \]
\end{lemma}
\begin{proof}
  Let $\rho \in \semel{\Gamma}$ and $w \in \WEL$ be such that $\semel{\Gamma \vd \Theta}_\rho (w)$.
  Let $m \in \MEL$.
  Suppose that $t, (h, \gamma) \vDash_w m \bullet r$ and $(\lambda x \ldotp \rho(e)), (h, \gamma) \xrightarrow{a} f, (h', \gamma')$.
  Then it must be that $a = \epsilon$, $f \in \FUN \setminus \dom{\gamma}$, $h' = h$ and $\gamma' = \gamma[f \mapsto \lambda x \ldotp \rho(e)]$.
  Let $w' = (\gamma', \pi_2(w)) = (\pi_1(w)[f \mapsto \lambda x \ldotp \rho(e)], \pi_2(w))$, so $w' \geq w$ and $t, (h, \gamma') \vDash_{w'} m \bullet r$.
  Moreover, $\semel{\Gamma \vd \Theta}_\rho (w')$ by monotonicity.
  By assumption, we have, for all $v \in \VAL$,
  $
    \semel{\Gamma, x : \ty{Val} \vd \triple{P}{e}{Q}}_{\rho{[x \mapsto v]}}(w')
  $.
  Consequently,
  \begin{align*}
&    m \in \semel{\Gamma \vd \lambda r \ldotp \forall x : \ty{Val} \ldotp \spec(\triple{P}{r x}{Q}) : \ty{Val} \to \ty{Spec}}_{\rho}(w')(f) \\
        & = \mathsf{wp}(\semel{\Gamma \vd \lambda r \ldotp \forall x : \ty{Val} \ldotp \spec(\triple{P}{r x}{Q}) : \ty{Val} \to \ty{Spec}}_{\rho})(f)(w')
  \end{align*}
  and so $\Gamma \mid \Theta \modelsel \triple{\emp}{\lambda x \ldotp e}{r \ldotp \forall x : \ty{Val} \ldotp \spec(\triple{P}{r x}{Q})}$ as required.
\end{proof}

\begin{lemma}[\textsc{SpecOut} rule]
  If
  \[
    \Gamma \mid \Theta, S \modelsel \triple{P}{e}{R}
  \]
  then
  \[
	  \Gamma \mid \Theta \modelsel \triple{P * \spec(S)}{e}{R}
  \]
\end{lemma}
\begin{proof}
  Let $\rho \in \semel{\Gamma}$ and $w \in \WEL$ be such that $\semel{\Gamma \vd \Theta}_\rho (w)$.
  Let $m \in \semel{\Gamma \vd P * \spec(S)}_\rho (w)$.
  Then $\sem{\Gamma \vd S}_\rho (w)$ and $m \in \semel{\Gamma \vd P}_\rho (w)$.
  Hence $\semel{\Gamma \vd \Theta, S}_\rho (w)$, and thus by assumption
  \[
    \semel{\Gamma \vd P : \ty{Prop}}_\rho(w) \subseteq \mathsf{wp}(\semel{\Gamma \vd Q : \ty{Val} \to \ty{Prop}}_\rho)(\rho(e))(w)
  \]
  It follows that $s \in \mathsf{wp}(\semel{\Gamma \vd Q : \ty{Val} \to \ty{Prop}}_\rho)(\rho(e))(w)$, as required.
\end{proof}

\begin{lemma}[\textsc{Alloc} rule]
  \[
    \Gamma, v : \ty{Val} \mid \Theta \modelsel \triple{\emp}{\refer v}{r \ldotp r \mapsto v}
  \]
\end{lemma}
\begin{proof}
  Let $\rho \in \semel{\Gamma, v : \ty{Val}}$ and $w \in \WEL$ be such that $\semel{\Gamma, v : \ty{Val} \vd \Theta}_\rho (w)$.
  Let $m \in \MEL$.
  Suppose that $t, (h, \gamma) \vDash_w m \bullet r$ and $\refer \rho(v), (h, \gamma) \xrightarrow{a} l, (h', \gamma')$.
  Then it must be that $a = \epsilon$, $l \in \Loc \setminus \dom{h}$, $h' = h[l \mapsto \rho(v)]$ and $\gamma' = \gamma$.
  Consequently, $t, (h', \gamma') \vDash_w (m \bullet [l \mapsto \rho(v)]) \bullet r$.
  Now 
  \begin{align*}
    m \bullet [x \mapsto& \rho(v)] \in \semel{\Gamma, v : \ty{Val} \vd \lambda r \ldotp r \mapsto v : \ty{Val} \to \ty{Prop}}_\rho (l)(w) \\
                                    &\quad = \mathsf{wp}(\semel{\Gamma, v : \ty{Val} \vd \lambda r \ldotp r \mapsto v : \ty{Val} \to \ty{Prop}}_\rho)(l)(w)
  \end{align*}
  as required.
\end{proof}

\begin{lemma}[\textsc{Read} rule]
  \[
    \Gamma, l, v : \ty{Val} \mid \Theta \modelsel \triple{l \mapsto v}{!l}{r \ldotp l \mapsto v * r = v}
  \]
\end{lemma}
\begin{proof}
  Let $\rho \in \semel{\Gamma, l, v : \ty{Val}}$ and $w \in \WEL$ be such that $\semel{\Gamma, l, v : \ty{Val} \vd \Theta}_\rho (w)$.
  Let $m \in \semel{l \mapsto v}_\rho(w)$.
  Suppose that $t, (h, \gamma) \vDash_w m \bullet r$ and $!\rho(l), (h, \gamma) \xrightarrow{a} e, (h', \gamma')$.
  (Note that since $h(\rho(l))) = (m \bullet r)(\rho(l)) = \rho(v)$, we do not have $!\rho(l), (h, \gamma) \nrightarrow$.)
  It must be that $a = \epsilon$, $e = \rho(v)$, $h' = h$ and $\gamma' = \gamma$.
  Now 
  \begin{align*}
&    m \in \semel{\Gamma, l, v : \ty{Val} \vd \lambda r \ldotp l \mapsto v * r = v : \ty{Val} \to \ty{Prop}}_\rho (\rho(v))(w) \\
        & = \mathsf{wp}(\semel{\Gamma, l, v : \ty{Val} \vd \lambda r \ldotp l \mapsto v * r = v : \ty{Val} \to \ty{Prop}}_\rho)(\rho(v))(w)
  \end{align*}
  as required.
\end{proof}

\begin{lemma}[\textsc{Write} rule]
  \[
    \Gamma, l, v : \ty{Val} \mid \Theta \modelsel \triple{l \mapsto \_}{l := v}{r \ldotp l \mapsto v * r = ()}
  \]
\end{lemma}
\begin{proof}
  Let $\rho \in \semel{\Gamma, l, v : \ty{Val}}$ and $w \in \WEL$ be such that $\semel{\Gamma, l, v : \ty{Val} \vd \Theta}_\rho (w)$.
  Let $m \in \semel{l \mapsto \_}_\rho(w)$.
  Suppose that $t, (h, \gamma) \vDash_w m \bullet r$ and $(\rho(l) := \rho(v)), (h, \gamma) \to e, (h', \gamma')$.
  (Note that since $h(\rho(l)) = m(\rho(l)) = v_0$ for some $v_0 \in \VAL$, we do not have $(\rho(l) := \rho(v)), (h,\gamma) \nrightarrow$.)
  It must be that $a = \epsilon$, $e = ()$, $h' = h[\rho(l) \mapsto \rho(v)]$ and $w' = w$.
  Let $m' = (h', \pi_2(m)) = (\pi_1(m)[\rho(l) \mapsto \rho(v)], \pi_2(m))$, so that $t, (h', \gamma) \vDash_w m'$
  Now 
  \begin{align*}
    m' &\in \semel{\Gamma, l, v : \ty{Val} \vd \lambda r \ldotp l \mapsto v * r = () : \ty{Val} \to \ty{Prop}}_\rho (())(w) \\
                                 & = \mathsf{wp}(\semel{\Gamma, l, v : \ty{Val} \vd \lambda r \ldotp l \mapsto v * r = () : \ty{Val} \to \ty{Prop}}_\rho)(())(w)
  \end{align*}
  as required.
\end{proof}

\begin{lemma}[\textsc{Emit} rule]
  Let $\Gamma' = \Gamma, t : \ty{seq}~\ty{Val}, v : \ty{Val}, I : \ty{seq}~\ty{Val} \to \ty{Bool}$.
  If
  \[
    \Gamma' \mid \Theta \modelsel t \cdot v \in I
  \]
  then
  \[
    \Gamma' \mid \Theta \modelsel \triple{\trace(t) * \inv(I)}{\emit~v}{r \ldotp \trace(t \cdot v) * r = ()}
  \]
\end{lemma}
\begin{proof}
  Let $\rho \in \semel{\Gamma'}$ and $w \in \WEL$ be such that $\semel{\Gamma' \vd \Theta}_\rho (w)$.
  Let $m \in \semel{\trace(t) * \inv(I)}_\rho(w)$.
  Suppose that $t_0, s \vDash_w m \bullet r$.
  It follows that $t_0 = \rho(t)$ and, by the assumption, $t_0 \cdot \rho(v) \in \pi_2(w)$.
  Suppose also that $\emit~\rho(v), s \xrightarrow{a} e, s'$.
  (Note that we do not have $\emit~\rho(v), (h, \gamma) \nrightarrow$.)
  It must be that $a = \rho(v)$, $e = ()$ and $s' = s$.
  Let $m' = (\pi_1(m), \pi_2(m) \cdot \rho(v)) = (\pi_1(m), t_0 \cdot a)$.
  It must also be that $t_0 \cdot a \vDash_w m' \bullet r$.
  Now 
  \begin{align*}
    m' &\in \semel{\Gamma' \vd \lambda r \ldotp \trace(t \cdot v) * r = () : \ty{Val} \to \ty{Prop}}_\rho (())(w) \\
       &\quad = \mathsf{wp}(\semel{\Gamma' \vd \lambda r \ldotp \trace(t \cdot v) * r = () : \ty{Val} \to \ty{Prop}}_\rho)(())(w)
  \end{align*}
  as required.
\end{proof}

\td{TODO: the entailment axioms}

\begingroup
\def\thetheorem{\ref{lem:el-multistep-exec}}
\begin{lemma}
Suppose that $\Gamma \mid - \vd \triple{P}{e}{Q}$
and let $w \in \WEL, \rho \in \semel{\Gamma}$ and $m \in \semel{\Gamma \vdash P}_\rho(w)$.
Taking $r, t, s$ such that $t, s \vDash_w m$, suppose that 
there exists $s',t'$ with $\rho(e),s \xrightarrow{t'}^*  e',s' \not\rightarrow$.
Then $e' \in \VAL$ and there exists $w' \geq w$ and $m' \in \semel{\Gamma \vdash Q}_\rho(e')(w')$
such that $(t \cdot t'), s' \vDash_{w'} m'$.
\end{lemma}
\addtocounter{theorem}{-1}
\endgroup

\begin{proof}
From Theorem~\ref{thm:soundel}, we get that
 $\Gamma \mid - \modelsel \triple{P}{e}{Q}$,
 i.e. 
 $\forall w\in\WEL.\forall \rho \in \semel{\Gamma}. \semel{\Gamma \vd \triple{P}{e}{Q} : \ty{Spec}}_\rho(w) = \top$.
 Thus, we get that
 $\semel{\Gamma \vdash P : \ty{Prop}}_\rho(w) \subseteq \mathsf{wp}(\semel{\Gamma \vdash Q : \ty{Val} \to \ty{Prop}}_\rho)(\rho(e))(w)$.
 So taking $m, r, t, s$ such that $m \in \semel{\Gamma \vdash P}_\rho(w)$ and $t, s \vDash_w m$,
 we know that $m \in \mathsf{wp}(\semel{\Gamma \vdash Q : \ty{Val} \to \ty{Prop}}_\rho)(\rho(e))(w)$.
 
 \noindent
 Writing $q$ for $\semel{\Gamma \vdash Q : \ty{Val} \to \ty{Prop}}_\rho$, from
 \begin{itemize}
  \item $m \in \mathsf{wp}(q)(\rho(e))(w)$ 
  \item $t, s \vDash_w m$
  \item $s, e \xrightarrow{t'}^* s', e'$
 \end{itemize}
 we prove by induction on the length of this reduction above that
  $e' \in \VAL$ and there exists $w' \geq w$ and $m' \in q(e')(w')$
such that $(t \cdot t'), s' \vDash_{w'} m'$.

For the base case, $s' = s$, $e' = \rho(e)$ and $t' = \varepsilon$. 
Unfolding the definition of \textsf{wp}, 
since $\rho(e),s \not\rightarrow$ it follows that $\rho(e) \in \VAL$ and
that 
  $m \in q(\rho(e))(w)$,
  so that we take $w' = w$ and $m' = m$.

For the inductive case, assume $\rho(e),s \xrightarrow{a} e'',s'' \xrightarrow{t'}^n e',s' \not\rightarrow$. 
Unfolding the definition of \textsf{wp}, it follows that there exists $w' \geq w$ and $m'$ such that
\begin{align*}
  t \cdot a, s'' \vDash_{w'} m' &&
  m' \in \mathsf{wp}(q)(e'')(w')
\end{align*}
By the induction hypothesis, it follows that there exists $w'' \geq w'$ and $m'' \in q(e')(w'')$
such that $t \cdot a \cdot t', s' \vDash_{w''} m''$.
\end{proof}


\section{Erasure of instrumentation}
\label{sec:erase-emit}

\begin{lemma}
For all $e,e',h,h',\gamma,\gamma'$,
if $\config{e,(h,\gamma)} \rightarrow \config{e',(h',\gamma')}$ (without emitting any events)
then
$\config{\erase{e},(h,\erase{\gamma})} \rightarrow \config{\erase{e'},(h',\erase{\gamma'})}$.
\end{lemma}

\begin{proof}
By case analysis on the reduction step,
using the fact that $\erase{K[e]} = \erase{K}[\erase{e}]$.
\end{proof}

\begingroup
\def\thetheorem{\ref{thm:erasure-emit}}
\begin{theorem}
For all $e,e',h,h',\gamma,\gamma'$,
if $\config{\erase{e},(h,\erase{\gamma})} \rightarrow^* \config{e',(h',\gamma')}$ 
then there exists a trace $t$, an expression $e''$ and an environment $\gamma''$ such that
$e' = \erase{e''}, \gamma' = \erase{\gamma''}$ and
$\config{e,(h,\gamma)} \xrightarrow{t}{\!\!}^*\ \config{e'',(h',\gamma'')}$.
\end{theorem}
\addtocounter{theorem}{-1}
\endgroup

\begin{proof}
 By induction on the (lexicographic order of the) length of the reduction of $\config{\erase{e},(h,\erase{\gamma})} \rightarrow^* \config{e',(h',\gamma')}$
 and the number of $\emit$s in $e$.
 If $\config{\erase{e},(h,\erase{\gamma})}$ is irreducible, then either:
 \begin{itemize}
  \item $\config{e,(h,\gamma)}$ is irreducible,
  so that we can define $t$ as the empty trace, $e''$ as $e$
  and $\gamma''$ as $\gamma$.
  \item or $e=K[\emit \, u]$ so that $\config{e,(h,\gamma)} \xrightarrow{u} \config{K[\unit],(h,\gamma)}$.
  Then we can apply the induction hypothesis since $K[\unit]$ has strictly fewer $\emit$s than $e$, and $\erase{e} = \erase{K[\unit]}$.
 \end{itemize}
Otherwise, there exist $e_1,\gamma_1,h_1$ such that
$\config{\erase{e},(\erase{\gamma},h)} \rightarrow \config{e_1,(\gamma_1,h_1)} \rightarrow^* \config{e',(\gamma',h')}$.
There are two possibilities:
\begin{itemize}
 \item Either there exist $e'_1,(\gamma'_1,h'_1)$, such that 
 $\config{e,(h,\gamma)} \rightarrow \config{e'_1,(\gamma'_1,h'_1)}$ without emitting any events, in which case,
 from the previous lemma and the determinism of the reduction, we get that 
 $\erase{e'_1} = e_1, \erase{\gamma'_1} = \gamma_1$ and $h'_1=h_1$.
 We now use the IH to prove the claim.
 \item Or $e=K[\emit \ u]$ so that $\config{e,(h,\gamma)} \xrightarrow{u} \config{K[\unit],(h,\gamma)}$.
 Then we can apply the IH since $K[\unit]$ has strictly fewer $\emit$s than $e$, and $\erase{e} = \erase{K[\unit]}$.
\end{itemize}
\end{proof}

\end{document}
